\documentclass[a4paper,12pt]{elsarticle}
\usepackage[T1]{fontenc}
\usepackage[english]{babel}
\usepackage{ae,aecompl}
\usepackage{amsmath, amsthm}
\usepackage{amssymb}
\usepackage[utf8]{inputenc}
\usepackage[section] {placeins}
\usepackage{hyperref}
\hypersetup{colorlinks=true} 
\usepackage[ruled, vlined, linesnumbered]{algorithm2e}
\newtheorem {Theorem}                 {Theorem}         [section]
\newtheorem {theorem}      [Theorem]  {Theorem}
\newtheorem {myalgorithm}    [Theorem]  {Algorithm}  

\newtheorem {Procedure}    [Theorem]  {Prozedur}    

\newtheorem {corollary}    [Theorem]  {Corollary}

\newtheorem {Lemma}        [Theorem]  {Lemma}
\newtheorem {lemma}        [Theorem]  {Lemma}

\input amssym.def
\input amssym

\usepackage{pgfplots}
\usepackage{verbatim}
\usepackage{tikz}
\usetikzlibrary{shapes,arrows,backgrounds,mindmap}
\usepackage{titlesec}
\journal{arXiv}

\begin{document}
\begin{frontmatter}
\title{Computing the $2$-blocks of directed graphs}
\author{Raed Jaberi\corref{cor1}}
\cortext[cor1]{Tel.: +49 3677 69 2786; fax: +49 3677 69 1237.}
\address{Faculty of Computer Science and Automation, Teschnische Universität Ilmenau, 
\\$98694$ Ilmenau, Germany}
\ead{raed.jaberi@tu-ilmenau.de}
\begin{abstract}  
Let $G$ be a directed graph. A \textit{$2$-directed block} in $G$ is a maximal vertex set $C^{2d}\subseteq V$ with $|C^{2d}|\geq 2$ such that for each pair of distinct vertices $x,y \in C^{2d}$, there exist two vertex-disjoint paths from $x$ to $y$ and two vertex-disjoint paths from $y$ to $x$ in $G$. In contrast to the $2$-vertex-connected components of $G$, the subgraphs induced by the $2$-directed blocks may consist of few or no edges. In this paper we present two algorithms for computing the $2$-directed blocks of $G$ in $O(\min\lbrace m,(t_{sap}+t_{sb})n\rbrace n)$ time, where $t_{sap}$ is the number of the strong articulation points of $G$ and $t_{sb}$ is the number of the strong bridges of $G$.  
Furthermore, we study two related concepts: the $2$-strong blocks and the $2$-edge blocks of $G$.
We give two algorithms for computing the $2$-strong blocks of $G$ in $O( \min \lbrace m,t_{sap} n\rbrace  n)$ time and we show that the $2$-edge blocks of $G$ can be computed in $O(\min \lbrace m, t_{sb} n \rbrace n)$ time. In this paper we also study some optimization problems related to the strong articulation points and the $2$-blocks of a directed graph. Given a strongly connected graph $G=(V,E)$, find a minimum cardinality set $E^{*}\subseteq E$ such that $G^{*}=(V,E^{*})$ is strongly connected and the strong articulation points of $G$ coincide with the strong articulation points of $G^{*}$. This problem is called minimum strongly connected spanning subgraph with the same strong articulation points. We show that there is a linear time $17/3$ approximation algorithm for this NP-hard problem. We also consider the problem of finding a minimum strongly connected spanning subgraph with the same $2$-blocks in a strongly connected graph $G$. We present approximation algorithms for three versions of this problem, depending on the type of $2$-blocks. 
\end{abstract} 
\end{frontmatter}
\section{Introduction}
Let $G=(V,E)$ be a directed graph with $|V|=n$ vertices and $|E|=m$ edges. A \textit{strong articulation point} (SAP) of $G$ is a vertex whose removal increases the number of strongly connected components (SCCs) of $G$. A strong bridge of $G$ is an edge whose removal increases the number of SCCs of $G$. We use $t_{sap}$ to denote the number of the strong articulation points (SAPs) of $G$ and $t_{sb}$ to denote the number of the strong bridges of $G$. A directed graph $G=(V,E)$ is said to be $k$-vertex-connected if it has at least $k+1$ vertices and the induced subgraph on $V\setminus X$ is strongly connected for every $X \subsetneq V $ with $|X|<k$. Thus, a strongly connected graph $G=(V,E)$ is $2$-vertex-connected if and only if it has at least $3$ vertices and it contains no SAPs. The \textit{$2$-vertex-connected components} of a strongly connected graph $G$ are its maximal $2$-vertex-connected subgraphs. The concept was defined in \cite{ES80}. For more information see \cite{ILS12}.

In $2010$, Georgiadis \cite{G10} gave a linear time algorithm to test whether a strongly connected graph $G$ is $2$-vertex-connected or not. Later, Italiano et al. \cite{ILS12} gave a linear time algorithm for the same problem which is faster in practice than the algorithm of Georgiadis \cite{G10}. Furthermore, Italiano et al. \cite{ILS12} presented a linear time algorithm for finding all the SAPs of a directed graph $G$. They also gave two linear time algorithms for calculating all the strong bridges of a directed graph $G$. In $2014$, Jaberi \cite{J14} presented algorithms for computing the $2$-vertex-connected components of directed graphs in $O(nm)$ time. The concept of $2$-vertex-connected components is not ideal because there are directed graphs in which many vertices are well connected with each other but they lie in distinct $2$-vertex-connected components or in no $2$-vertex-connected component. In this paper we study alternative concepts similar to the $k$-blocks of undirected graphs which were defined in \cite{CDHH13} as follows. A \textit{$k$-block} in an undirected graph $G=(V,E)$ is a maximal vertex set $U\subseteq V$ with $|U|\geq k$ such that no set $X\subseteq V$ with $|X|<k$ separates any two vertices of $U\setminus X$ in the undirected graph $G$. In $2013$, Carmesin et al. \cite{CDHH13} showed that there exists a $O(\min\lbrace k,\sqrt{n}\rbrace n^{4})$-time algorithm that calculates all the \textit{$k$-blocks} in an undirected graph. The $2$-blocks in an undirected graph $G$ are similar to the $2$-vertex connected components of the undirected graph $G$, which can be found in linear time using Tarjan's algorithm \cite{T72}. In this paper we introduce and study three new concepts: the $2$-directed blocks, the $2$-strong blocks, and the $2$-edge blocks of directed graphs.
A \textit{$2$-directed block} in $G$ is a maximal vertex set $C^{2d}\subseteq V$ with $|C^{2d}|\geq 2$ such that for each pair of distinct vertices $x,y \in C^{2d}$, there exist two vertex-disjoint paths from $x$ to $y$ and two vertex-disjoint paths from $y$ to $x$ in $G$. In contrast to the $2$-vertex-connected components of $G$, the subgraphs induced by the $2$-directed blocks may consist of few or no edges. A \textit{$2$-strong block} in $G$ is a maximal vertex set $C^{2s}\subseteq V$ with $|C^{2s}|\geq 2$ such that for each pair of distinct vertices $x,y \in C^{2s}$ and for each vertex $z \in V\setminus \lbrace x,y\rbrace$, the vertices $x$ and $y$ lie in the same strongly connected component (SCC) of the graph $G\setminus \lbrace z\rbrace$. 
A \textit{$2$-edge block} in $G$ is a maximal vertex set $C^{2e}\subseteq V$ with $|C^{2e}|>1$ such that for each pair of distinct vertices $v, w \in C^{2e}$, there are two edge-disjoint paths from $v$ to $w$ and two edge-disjoint paths from $w$ to $v$ in $G$. These concepts capture the idea that it is difficult to separate vertices in a block in slightly different ways, and very different from the concept of $2$-vertex-connected components. Our new concepts are illustrated in Figure \ref{fig:example2Blocksand2Vcc}. 
   \begin{figure}[htbp]
    \centering
    \scalebox{0.95}{
    \begin{tikzpicture}[xscale=2]
        \tikzstyle{every node}=[color=black,draw,circle,minimum size=0.8cm]
       \node (v1) at (-0.5,2) {$1$};
        \node (v2) at (-2.5,0) {$2$};
        \node (v3) at (-0.5, -2.5) {$3$};
        \node (v4) at (1,-0.5) {$4$};
        \node (v5) at (0.7,2.2) {$5$};
        \node (v6) at (2.5,0.4) {$6$};
        \node (v7) at (2.5,2.5) {$7$};
        \node (v8) at (0.5,-2.8) {$8$};
        \node (v9) at (1.8,-0.8) {$9$};
        \node (v10) at (2.2,-2) {$10$};
        \node (v11) at (3.5,-2.5) {$11$};
        \node (v12) at (3.5,2) {$12$};
       \begin{scope}   
            \tikzstyle{every node}=[auto=right]   
             \draw [-triangle 45] (v1) to (v2);
             \draw [-triangle 45] (v2) to [bend left ](v1);
             \draw [-triangle 45] (v1) to (v3);
             \draw [-triangle 45] (v3) to [bend left ](v1);
             \draw [-triangle 45] (v2) to [bend right ] (v3);
             \draw [-triangle 45] (v3) to (v2);
             \draw [-triangle 45] (v3) to (v4);
             \draw [-triangle 45] (v5) to (v3);
             \draw [-triangle 45] (v5) to (v7);
             \draw [-triangle 45] (v6) to (v5);
             \draw [-triangle 45] (v8) to (v4);
             \draw [-triangle 45] (v9) to (v8);
             \draw [-triangle 45] (v6) to (v9);
             \draw [-triangle 45] (v8) to (v10);
             \draw [-triangle 45] (v10) to [bend right ] (v6);
             \draw [-triangle 45] (v6) to [bend left ] (v11);
             \draw [-triangle 45] (v11) to [bend left ] (v8);
             \draw [-triangle 45] (v6) to (v12);
             \draw [-triangle 45] (v12) to (v7);
             \draw [-triangle 45] (v7) to (v6);
             \draw [-triangle 45] (v4) to (v6);
             \draw [-triangle 45] (v4) to[bend right ]  (v10);
             \draw [-triangle 45] (v6) to (v1);
             \draw [-triangle 45] (v10) to (v11);
             \draw [-triangle 45] (v2) to (v6);
             \draw [-triangle 45] (v10) to (v9);
        \end{scope}
    \end{tikzpicture}}
     \caption{A strongly connected graph $G$, which contains one $2$-vertex-connected component $\lbrace1,2,3\rbrace$, two $2$-directed blocks $\lbrace 6,1,2,3\rbrace,\lbrace 8,10,6,4\rbrace$, four $2$-strong blocks $\lbrace 6,1,2,3\rbrace,\lbrace 9,8\rbrace,\lbrace 8,10,6,4\rbrace,\lbrace 7,6\rbrace$, and one $2$-edge block $\lbrace 1,2,3,4,6,8,10\rbrace$. Notice that the $2$-vertex-connected component $\lbrace1,2,3\rbrace$ is a subset of the $2$-directed block $\lbrace 6,1,2,3\rbrace$. We shall also see that each $2$-directed block is a subset of a $2$-strong block.}
\label{fig:example2Blocksand2Vcc}
\end{figure}

In this paper we also study some optimization problems related to the SAPs and the $2$-blocks of a directed graph. Given a strongly connected graph $G=(V,E)$, find a minimum cardinality set $E^{*}\subseteq E$ such that $G^{*}=(V,E^{*})$ is strongly connected and the SAPs of $G$ coincide with the SAPs of $G^{*}$. This problem is called minimum strongly connected spanning subgraph (MSCSS) with the same SAPs, denoted by MS-SAPs. Moreover, We consider the problem of finding a MSCSS with the same $2$-blocks, defined as follows. Given a strongly connected graph $G=(V,E)$, the goal is to find a subset $E^{*}\subseteq E$ of minimum size such that $G^{*}=(V,E^{*})$ is strongly connected and the $2$-blocks of $G$ coincide with $2$-blocks of $G^{*}=(V,E^{*})$. There are three versions of this problem, depending on the type of $2$-blocks: MSCSS with the same $2$-directed blocks (denoted by MS-2DBs), MSCSS with the same $2$-strong blocks (denoted by MS-2SBs), and MSCSS with the same $2$-edge blocks (denoted by MS-2EBs). MS-2DBs and MS-2SBs problems correspond to the problem of finding a minimum-size $2$-vertex-connected spanning subgraph of an undirected graph \cite{VV00}.
The problem of finding a minimum-size $2$-vertex connected spanning subgraph of a directed graph $G$ is a special case of the problems MS-SAPs, MS-2SBs and MS-2DBs when $G$ is $2$-vertex-connected. Moreover, the problem of finding a minimum-cardinality $2$-edge connected spanning subgraph of a directed graph $G$ is a special case of the MS-2EBs problem when $G$ is $2$-edge-connected. Therefore, by results from \cite{CJ79} the problems MS-SAPs, MS-2SBs MS-2EBs and MS-2DBs are NP hard. 

Let $G$ be a directed graph. In this paper, we present two algorithms for computing the $2$-directed blocks of $G$ in $O(\min\lbrace m,(t_{sap}+t_{sb})n\rbrace n)$ time. We also present two algorithms for computing the $2$-strong blocks of $G$ in $O( \min \lbrace m,t_{sap} n\rbrace  n)$ time and we show that the $2$-edge blocks of $G$ can be computed in $O(\min \lbrace m, t_{sb} n \rbrace n)$ time.
Furthermore, we show that there is a linear time $17/3$ approximation algorithm for the MS-SAPs problem. We also present a $(2t_{sap}+17/3)$ approximation algorithm for the MS-2SBs problem and a $(2t_{sb} +4)$ approximation algorithm for the MS-2EBs problem. Moreover, we prove that there exist a $(2(t_{sap}+ t_{sb}) +29/3)$ approximation algorithm for the MS-2DBs problem. 
\subsection{Related Work}
In independent work, Georgiadis, Italiano, Laura, and Parotsidis \cite{GILP14} have studied $2$-edge blocks and have given linear time algorithms for finding them. This is better than our results in Section \ref{def:sec2eb}.
\section{Graph Terminology and Notation} \label{def:DandN}
Throughout this paper we consider only simple directed graphs. In this section we recall some basic definitions \cite{LT79,LM69,ILS12}.
A \textit{flowgraph} $G(v)=(V,E,v)$ is a directed graph with $|V|=n$ vertices, $|E|=m$ edges, and a distinguished start vertex $v \in V$ such that every vertex $w\in V$ is reachable from $v$. For a flowgraph $G(v)=(V,E,v)$, the \textit{dominance relation} of $G(v)$ is defined as follows: a vertex $u \in V$ is a \textit{dominator} of vertex $w \in V$ if every path from $v$ to $w$ includes $u$. By $dom(w)$ we denote the set of dominators of vertex $w$. Obviously, the set of dominators of the start vertex in $G(v)$ is $dom(v)=\lbrace v \rbrace$. For every vertex $w\in V$ with $w \neq v$, $\lbrace v,w \rbrace$ is a subset of $dom(w)$; we call $w,v$ the \textit{trivial dominators} of $w$. A vertex $u$ is a \textit{non-trivial dominator} in $G(v)$ if there is some $w \notin \lbrace v,u\rbrace$ such that $u \in dom(w)\setminus\lbrace v\rbrace$. The set of all non-trivial dominators is denoted by $D(v)$.
The dominance relation is transitive. A vertex $u\in V$ is an \textit{immediate dominator} of vertex $w \in V\setminus\lbrace v\rbrace$ in $G(v)$ if $u\in dom(w)\setminus\lbrace w\rbrace$ and all elements of $dom(w)\setminus \lbrace w \rbrace$ are dominators of $u$. Every vertex $w$ of $G(v)$ except the start vertex $v$ has a unique immediate dominator. The edges $(u,w)$ where $u$ is the immediate dominator of $w$ form a tree with root $v$, called the \textit{dominator tree} of $G(v)$, denoted by $DT(v)$. Two spanning trees $T$ and $T'$ of $G(v)$ are called independent if for every vertex $w\in V\setminus \lbrace v\rbrace$, the paths from $v$ to $w$ in $T$ and $T'$ contain only $dom(w)$ in common \cite{GT12}. An edge $(x,y)$ is an \textit{edge dominator} of vertex $w$ if every path from  $v$ to $w$ in $G(v)$ contains edge $(x,y)$. Let $G=(V,E)$ be a directed graph. Let $F$ be a subset of $E$ and let $U$ be a subset of $V$. We use 
$G\setminus F$ to denote the directed graph obtained from $G$ by deleting all the edges in $F$. We use $G\setminus U$ to denote the directed graph obtained form $G$ by removing all the vertices in $U$ and their incident edges. The reversal graph of $G$ is the directed graph 
$G^{R}=(V,E^{R})$, where $E^{R}=\lbrace (w,u) \mid (u,w) \in E \rbrace$. Let $v$ be a vertex in $G$. By $D^{R}(v)$ we denote the set of all non-trivial dominators in the flowgraph $G^{R}(v)=(V,E^{R},v)$. Let $G=(V,E)$ be an undirected graph. A block of $G$ is a maximal connected subgraph of $G$ that contains no articulation points. An undirected graph $G$ is called chordal if every cycle of length at least $4$ has a chord \cite{G72,RT75}.
\section{Computing $2$-directed blocks} \label{def:c2dbofdg} 
 In this section we present our first algorithm for computing the $2$-directed blocks of directed graphs. Our second algorithm will be described in section \ref{def:secanorithAlgorithmfor2db}. We consider only strongly connected graphs since the $2$-directed blocks of a directed graph are the union of the $2$-directed blocks of its SCCs.
Let $G=(V,E)$ be a strongly connected graph. We write $x \overset{2}{\rightsquigarrow} y$ if there exist two vertex-disjoint paths from $x$ to $y$ in $G$. For distinct vertices $x,y \in V$ we write
$x \overset{2}{\leftrightsquigarrow } y$ if $x \overset{2}{\rightsquigarrow} y$ and $y \overset{2}{\rightsquigarrow} x$. A $2$-directed block in $G$ is a maximal vertex set $C^{2d}\subseteq V$ with $|C^{2d}|\geq 2$ such that for each pair of distinct vertices $x,y \in C^{2d}$, we have $x \overset{2}{\leftrightsquigarrow } y$.
\begin{lemma} \label{def:lem1}
Let $G=(V,E)$ be a strongly connected graph and let $x,y$ be distinct vertices in $G$. Then $x \overset{2}{\leftrightsquigarrow } y$ if and only if for each vertex $w\in V\setminus \lbrace x,y \rbrace$ the vertices $x,y$ lie in the same SCC of $G\setminus \lbrace w\rbrace$ and in the same SCC of $G\setminus \lbrace (x,y),(y,x)\rbrace$.
\end{lemma}
\begin{proof} ``$\Leftarrow$'': It is sufficient to show that there are two vertex-disjoint paths from $x$ to $y$ in $G$. We consider two cases.
\begin{enumerate}
\item $(x,y)\notin E$. Let $w$ be any vertex in $V\setminus \lbrace x,y \rbrace$. Since the vertices $x,y$ lie in the same SCC of $G\setminus \lbrace w\rbrace$, there exists a path from $x$ to $y$ in $G\setminus \lbrace w\rbrace$. Thus,
one can not interrupt all paths from $x$ to $y$ by removing $w$ from $G$.
Since $x$ and $y$ are not adjacent, by Menger's Theorem for vertex connectivity \cite{BR12} we have $x \overset{2}{\rightsquigarrow } y$.
\item $(x,y) \in E$. Since $x,y$ lie in the same SCC of $G\setminus \lbrace (x,y),(y,x)\rbrace$, there is a path $p_1$ from $x$ to $y$ in $G\setminus \lbrace (x,y)\rbrace$. Thus, there are two vertex-disjoint paths $p_1$ and $p_2=(x,y)$ from $x$ to $y$ in $G$.
\end{enumerate}
``$\Rightarrow$'': We know there are two vertex-disjoint paths $p_1$ and $p_2$ from $x$ to $y$ in $G$. We must show that in $G\setminus \lbrace w\rbrace$ and in $G\setminus \lbrace (x,y)\rbrace$ there is a path from $x$ to $y$. Since at most one of $p_1$ and $p_2$ contains $w$ and at most one of $p_1$ and $p_2$ is edge $(x,y)$, the claim follows. 

\end{proof}
\begin{lemma}\label{def:xyInSameSccWithoutE}
Let $G=(V,E)$ be a strongly connected graph and let $x,y$ be distinct vertices in $G$ such that $x \overset{2}{\leftrightsquigarrow } y$. Then the vertices $x,y$ lie in the same SCC of $G\setminus\lbrace e\rbrace$ for any edge $e\in E$. 
\end{lemma}
\begin{proof}
There exist two vertex-disjoint paths $p_1,p_2$ from $x$ to $y$ and two vertex-disjoint paths $p_3,p_4$ from $y$ to $x$ in $G$ since $x \overset{2}{\leftrightsquigarrow } y$. The paths $p_1,p_2$ are edge-disjoint and the paths $p_3,p_4$ are also edge-disjoint. Hence, there exist a path from $x$ to $y$ and a path from $y$ to $x$ in $G\setminus\lbrace e\rbrace$ for any edge $e\in E$.
\end{proof}
$2$-directed blocks intersect in at most one vertex. ($2$-vertex-connected components have the same property, see \cite{ES80} and \cite{J14}.)
\begin{lemma}\label{def:lem2}
Let $C^{2d}_1,C^{2d}_2$ be distinct $2$-directed blocks in a strongly connected graph $G=(V,E)$. Then $C^{2d}_1$ and $C^{2d}_2$ have at most one vertex in common.
\end{lemma}
\begin{proof}
Indirect. Assume that $|C^{2d}_1 \cap C^{2d}_2|>1$. By renaming we can assume that there are at least two vertices $v \in C^{2d}_1,w\in C^{2d}_2$ with $v,w \notin C^{2d}_1 \cap C^{2d}_2$ such that there are no two vertex-disjoint paths from $v$ to $w$ in $G$. We consider two cases.
\begin{enumerate}
\item $(v,w) \notin E$. By Menger's Theorem \cite{BR12} there is some vertex $s \in V \setminus \lbrace v,w\rbrace$ such that $s$ lies on all paths from $v$ to $w$. Let $z$ be a vertex in $(C^{2d}_1 \cap C^{2d}_2)\setminus \lbrace s\rbrace$. Since $C^{2d}_1$ and $C^{2d}_2$ are $2$-directed blocks, there is a path from $v$ to $z$ in $G\setminus\lbrace s\rbrace $ and a path from $z$ to $w$ in $G\setminus\lbrace s\rbrace $, hence there is a path from $v$ to $w$ in $G\setminus\lbrace s\rbrace $, which is a contradiction. 
\item $(v,w) \in E$. In this case there is no path from $v$ to $w$ in $G\setminus \lbrace (v,w) \rbrace$. Let $u$ be a vertex in $C^{2d}_1 \cap C^{2d}_2$. But, again by the definition of $2$-directed blocks, there are paths from $v$ to $u$ and from $u$ to $w$ in $G\setminus \lbrace (v,w) \rbrace$, a contradiction.
\end{enumerate}
\end{proof}
Next we note that $2$-directed blocks can not form cycles in the following sense.
\begin{lemma}\label{def:AllVerticesOfCycleAreIn2directedBlock}
Let $G=(V,E)$ be a strongly connected graph and let $v_0,v_1,\ldots,v_l$ be distinct vertices of $G$ such that $v_0 \overset{2}{\leftrightsquigarrow } v_l$ and $v_{i-1} \overset{2}{\leftrightsquigarrow } v_i$ for $i \in\lbrace 1,2\ldots,l\rbrace$. Then all the vertices $v_0,v_1,\ldots,v_l$ lie in the same $2$-directed block of $G$.
\end{lemma}
\begin{proof}
Indirect. Assume that there exist two vertices $v_r,v_q$ with $r,q\in \lbrace 0,1,\ldots,l\rbrace$ such that $v_r,v_q$ lie in distinct $2$-directed blocks of $G$ and $r<q$. By renaming, we may assume that there do not exist two vertex-disjoint paths from $v_r$ to $v_q$ in $G$. We consider two cases.
\begin{enumerate}
\item $(v_r,v_q)\notin E$. In this case, all the paths from $v_r$ to $v_q$ contain a vertex $s\in V\setminus\lbrace v_r,v_q \rbrace$. Therefore, there is no path from $v_r$ to $v_q$ in $G\setminus \lbrace s\rbrace$. There are two cases to consider.
\begin{enumerate}
\item $s\notin \lbrace v_{r+1},v_{r+2},\ldots,v_{q-1}\rbrace$. In this case, for each $i\in \lbrace r+1,r+2,\ldots,q\rbrace$, there is a path from $v_{i-1}$ to $v_i$ in $G\setminus \lbrace s\rbrace$ by Lemma \ref{def:lem1}, a contradiction.
\item $s \in \lbrace v_{r+1},v_{r+2},\ldots,v_{q-1}\rbrace$. Then by Lemma \ref{def:lem1}, there are paths from $v_r$ to $v_{r-1},\ldots$ from $v_1$ to $v_0,$ from $v_0$ to $v_l,$ from $v_l$ to $v_{l-1},\ldots$ from $v_{q+1}$ to $v_q$ in $G\setminus \lbrace s\rbrace$, again a contradiction.
\end{enumerate}
\item $(v_r,v_q) \in E$. By Lemma \ref{def:xyInSameSccWithoutE}, for each $i\in \lbrace r+1,r+2,\ldots,q\rbrace$, the vertices $v_{i-1},v_i$ lie in the same SCC of $G\setminus\lbrace(v_r,v_q)\rbrace $. Therefore, there exists a path $p_1$ from $v_r$ to $v_q$ in $G\setminus\lbrace(v_r,v_q)\rbrace $. Consequently, there are two vertex-disjoint paths $p_1$ and $p_2=(v_r,v_q)$ from $v_r$ to $v_q$ in $G$, but this is a contradiction.
\end{enumerate}
\end{proof}
We construct the $2$-directed block graph $G^{2d}=(V^{2d},E^{2d})$ of a strongly connected graph $G=(V,E)$  as follows. It has a vertex $v_i$ for every $2$-directed block $C^{2d}_i$ and all vertices $w$ that lie in the intersection of (at least) two $2$-directed blocks. For each pair of distinct $2$-directed blocks $C^{2d}_i,C^{2d}_j$ with $C^{2d}_i \cap C^{2d}_j=\lbrace w\rbrace$, we add two undirected edges $(v_i,w),(w,v_j)$ to $E^{2d}$.
\begin{lemma}\label{def:G2dforest}
Let $G=(V,E)$ be a strongly connected graph. Then the $2$-directed block graph $G^{2d}=(V^{2d},E^{2d})$ of $G$ is a forest.
\end{lemma}
\begin{proof} This  follows from Lemma \ref{def:AllVerticesOfCycleAreIn2directedBlock}.
\end{proof}
Now we turn to algorithm for finding the $2$-directed blocks.
Algorithm \ref{algo:algorforall2dblocks} describes our first algorithm for computing all the $2$-directed blocks of a strongly connected graph $G$.
\begin{figure}[htbp]
\begin{myalgorithm}\label{algo:algorforall2dblocks}\rm\quad\\[-5ex]
\begin{tabbing}
\quad\quad\=\quad\=\quad\=\quad\=\quad\=\quad\=\quad\=\quad\=\quad\=\kill
\textbf{Input:} A strongly connected graph $G=(V,E)$.\\
\textbf{Output:} The $2$-directed blocks of $G$.\\
{\small 1}\> \textbf{if} $G$ is $2$-vertex-connected \textbf{then}\\
{\small 2}\>\> Output $V$.\\
{\small 3}\> \textbf{else}\\
{\small 4}\>\> Let $A$ be an $n\times n$ matrix.\\
{\small 5}\>\> Initialize $A$ with $0$s.\\
{\small 6}\>\> \textbf{for} each ordered pair $(v,w)\in V\times V$ \textbf{do}\\
{\small 7}\>\>\> \textbf{if} there are two vertex-disjoint paths from $v$ to $w$ in $G$ \textbf{then}\\
{\small 8}\>\>\>\> $A[v,w] \leftarrow 1$. \\
{\small 9}\>\> Construct undirected graph $G^{*}=(V,E^{*})$ as follows. \\
{\small 10}\>\> \textbf{for} each pair $(v,w) \in V\times V$ \textbf{do} \\
{\small 11}\>\>\> \textbf{if} $A[v,w]=1$ and $A[w,v]=1$ \textbf{then} \\
{\small 12}\>\>\>\> Add the undirected edge $(v,w)$ to $E^{*}$. \\
{\small 13}\>\> Compute the blocks of size $>1$ of $G^{*}=(V,E^{*})$ and output them. 
\end{tabbing}
\end{myalgorithm}
\end{figure}
\begin{lemma}\label{def:Algo1IsCorrect}
Algorithm \ref{algo:algorforall2dblocks} calculates $2$-directed blocks.
\end{lemma}
\begin{proof}
If $G$ is $2$-vertex connected, then $V$ is a $2$-directed block. Let $G=(V,E)$ be a strongly connected graph which is not $2$-vertex connected. For any vertices $v,w \in V$, $v \overset{2}{\leftrightsquigarrow } w$ if and only if $A[v,w]=1$ and $A[w,v]=1$ in line $11$. Hence, $v \overset{2}{\leftrightsquigarrow } w$ if and only if $(v,w) \in E^{*}$. Let $x,y$ be two vertices that do not lie in the same block of $G^{*}$.
Then $(x,y)$ can not be in $E^{*}$. Hence, the vertices $x,y$ do not lie in the same $2$-directed block of $G$. Let $B$ be a block of $G^{*}$ containing $v,w$ with $(v,w)\in E^{*}$. There are two cases to consider.
\begin{enumerate}
\item $B=\lbrace v,w\rbrace$. Then $v \overset{2}{\leftrightsquigarrow } w$ and $\lbrace v,w\rbrace$ is a $2$-directed block. (If there were some $z$ such that $v,w,z$ are in the same $2$-directed block, we would have the triangle $(v,z),(z,w),(w,v)$ in $G^{*}$, hence $z$ would be in the same block as $v,w$.) 
\item $B$ contains other vertices. We show that all these vertices are in the same $2$-directed block. If $z\in V\setminus\lbrace v,w\rbrace$ is in $B$, then $z,v$ lie on one simple cycle in $G^{*}$. By Lemma \ref{def:AllVerticesOfCycleAreIn2directedBlock}, the vertices $z,v$ lie in the same $2$-directed block.
\end{enumerate}
\end{proof}
It remains to describe Procedure \ref{algo:proccheck2vdps} that implements steps $6$--$8$ of Algorithm \ref{algo:algorforall2dblocks}.

\begin{figure}[htbp]
\begin{Procedure}\label{algo:proccheck2vdps}\rm\quad\\[-5ex]
\begin{tabbing}
\quad\quad\=\quad\=\quad\=\quad\=\quad\=\quad\=\quad\=\quad\=\quad\=\kill
\textbf{Purpose:} Check if there are two vertex disjoint paths.\\
\textbf{Input:} A strongly connected graph $G=(V,E)$.\\
\textbf{Output:} Matrix $A$.\\
{\small 1}\> \textbf{for} each vertex $v \in V$ \textbf{do} \\
{\small 2}\>\> $E'\leftarrow E$. \\
{\small 3}\>\> $V'\leftarrow V$. \\
{\small 4}\>\> \textbf{for} each edge $e=(v,w)\in E$ \textbf{do}\\
{\small 5}\>\>\> $E'\leftarrow E'\setminus \lbrace (v,w)\rbrace$.\\
{\small 6}\>\>\> $V'\leftarrow V'\cup \lbrace u_e\rbrace$.\\
{\small 7}\>\>\> $E'\leftarrow E' \cup \lbrace (v,u_e),(u_e,w) \rbrace$.\\
{\small 8}\>\> Compute the dominator tree $DT'(v)$ of the flowgraph $G'(v)=(V',E',v)$.\\
{\small 9}\>\> \textbf{for} each direct successor $w$ of $v$ in $DT'(v)$ \textbf{do}\\ 
{\small 10}\>\>\> \textbf{if} $w\in V$ \textbf{then}\\
{\small 11}\>\>\>\> $A[v,w] \leftarrow 1$.
\end{tabbing}
\end{Procedure}
\end{figure} 
For each vertex $v \in V$, we construct a directed graph $G'=(V',E')$ from $G$ as follows. For each edge $(v,w)\in E$, we remove this edge $(v,w)$ and we add a new vertex $u_e$ and two new edges $(v,u_e),(u_e,w)$ to $G'$. Then we compute the dominator tree $DT'(v)$ of the flowgraph $G'(v)=(V',E',v)$. For each direct successor $w$ of $v$ in $DT'(v)$ such that $w\in V$, line $11$ sets $A[v,w]$ to $1$. The correctness of Procedure \ref{algo:proccheck2vdps} follows from the following lemma.
\begin{lemma} \label{def:CorrectnessOfproccheck2vdps}
let $G=(V,E)$ be a strongly connected graph and let $v,w$ be two distinct vertices in $G$. Then $v \overset{2}{\rightsquigarrow} w$ in $G(v)$ if and only if $v$ is the immediate dominator of $w$ in the flowgraph $G'(v)=(V',E',v)$.
\end{lemma}
\begin{proof}
``$\Rightarrow$'' Assume that $v \overset{2}{\rightsquigarrow} w$ in $G(v)$. Then there are two vertex-disjoint paths $p_{1}=(v=v_{1},v_{2},\ldots,v_{t}=w)$ and $p_{2}=(v=u_{1},u_{2},\ldots,u_{l}=w)$ from $v$ to $w$ in $G(v)$. In lines $4$--$7$ of Procedure \ref{algo:proccheck2vdps}, the edge $x=(v_{1},v_{2})$ is replaced by two edges $(v_{1},v_{x}),(v_{x},v_{2})$ and the edge $y=(u_{1},u_{2})$ is replaced by two edges $(u_{1},u_{y}),(u_{y},u_{2})$. Since $v_{x}\neq u_{y}$, there exist two vertex-disjoint paths from $v$ to $w$ in $G'(v)$. Therefore, $v$ is the immediate dominator of $w$ in the flowgraph $G'(v)$.\\ 
``$\Leftarrow$'' Suppose that all paths from $v$ to $w$ in $G(v)$ contain at least vertex $x\in V\setminus\lbrace v,w\rbrace$. Then $x$ is a non-trivial dominator of $w$ in $G(v)$. Thus, $v$ is not the immediate dominator of $w$ in $G(v)$. Let $p=(v=v_{1},v_{2},\ldots,v_{t}=w)$ be a simple path from $v$ to $w$ in $G(v)$. In lines $4$--$7$ of Procedure \ref{algo:proccheck2vdps}, $e=(v_{1},v_{2})$ is replaced by $(v_{1},u_{e}),(u_{e},v_{2})$. Hence, the path $p$ corresponds to the simple path $(v=v_{1},u_{e},v_{2},\ldots,v_{t}=w)$ in $G'(v)$. Since $u_{e}\neq x$, the vertex $x$ is a non-trivial dominator of $w$ in $G'(v)$. Therefore, $v$ is not the immediate dominator of $w$ in $G'(v)$.  
\end{proof}
\begin{theorem}\label{def:algo2dbrunningtime}
Algorithm \ref{algo:algorforall2dblocks} runs in $O(nm)$ time.
\end{theorem}  
\begin{proof}
The dominators of a flowgraph can be found in linear time \cite{BGKRTW00,AHLT99}. Since the graph $G'$ has $|V'|=n+d_{out}(v)<2n$ vertices and $|E'|=m+d_{out}(v)< m+n$ edges, lines $2$--$11$ of Procedure \ref{algo:proccheck2vdps} take $O(n+m)$ time. Thus, the running time of Procedure \ref{algo:proccheck2vdps} is $O(nm)$. One can test whether a directed graph is $2$-vertex-connected in linear time using the algorithm of Italiano
et al. \cite{ILS12}. The initialization of matrix $A$ requires $O(n^{2})$ time.
 The undirected graph $G^{*}=(V,E^{*})$ can also be constructed in $O(n^{2})$ time. Furthermore, the blocks of $G^{*}$ can be computed in $O(n^{2})$ using Tarjan's algorithm \cite{T72}. The total cost is therefore $O(nm+n^{2})=O(nm)$.
\end{proof}
Let $G=(V,E)$ be a strongly connected graph. By definition, the $2$-directed blocks of $G$ are the maximal cliques of the auxiliary graph $G^{*}$ which is constructed in lines $4$--$12$ of Algorithm \ref{algo:algorforall2dblocks}. By Lemma \ref{def:AllVerticesOfCycleAreIn2directedBlock}, the auxiliary graph $G^{*}$ is chordal. In line $13$ of Algorithm \ref{algo:algorforall2dblocks}, one can compute the maximal cliques of the auxiliary graph $G^{*}$ instead of blocks since the maximal cliques of a chordal graph can be calculated in linear time \cite{G72,RT75}.
\section{Computing $2$-strong blocks} \label{def:secC2sb}
In this section we present two algorithms for computing the $2$-strong blocks of directed graphs. The $2$-strong blocks of a directed graph are the union of the $2$-strong blocks of its SCCs. Let $G=(V,E)$ be a strongly connected graph. We define a relation $\overset{2s}{\leftrightsquigarrow }$ as follows. For any distinct vertices $x,y\in V$, we write $x \overset{2s}{\leftrightsquigarrow } y$ if for any vertex $z\in V\setminus\lbrace x,y\rbrace$, the vertices $x,y$ lie in the same SCC of $G\setminus\lbrace z\rbrace$. By definition, the $2$-strong blocks are maximal subsets of $V$ of size $\geq2$ closed under $ \overset{2s}{\leftrightsquigarrow } $. 
Let $v,w$ be distinct vertices in $V$ such that $(v,w)\in E$ and $w \overset{2}{\rightsquigarrow} v$. While $v,w$ are in one $2$-strong block, these vertices do not necessarily lie in the same $2$-directed block of $G$.
\begin{lemma}\label{def:Each2dbIsIn2sb}
Each $2$-directed block in a strongly connected graph $G$ is a subset of a $2$-strong block in $G$. 
\end{lemma}
\begin{proof}
Immediate from Lemma \ref{def:lem1}.
\end{proof}
\begin{lemma}\label{def:2sbsHasAtMostOneVertex}
Let $G=(V,E)$ be a strongly connected graph. Let $C^{2s}_1,C^{2s}_2$ be distinct $2$-strong blocks in $G$. Then $C^{2s}_1$ and $C^{2s}_2$ have at most one vertex in common.
\end{lemma}
\begin{proof}
Indirect. Assume that $|C^{2s}_1 \cap C^{2s}_2|>1$. Then there exist at least two vertices $x \in C^{2s}_1,y\in C^{2s}_2$ with $x,y \notin C^{2s}_1 \cap C^{2s}_2$ and a vertex $z\in V\setminus \lbrace x,y\rbrace$ such that the vertices $x,y$ lie in different SCCs of $G\setminus \lbrace z\rbrace$. Let $w$ be a vertex in $(C^{2s}_1 \cap C^{2s}_2)\setminus \lbrace z\rbrace$. Since $x,w\in C^{2s}_1  $, these vertices lie in the same SCC of $G\setminus \lbrace z\rbrace$, similarly for $w,y \in C^{2s}_2 $. Hence $x,y$ lie in the same SCC of $G\setminus \lbrace z\rbrace$, a contradiction.
\end{proof}
As with $2$-directed blocks, there can not be cycles of $2$-strong blocks. The proof is even simpler.
\begin{lemma}\label{def:AllVerticesOfCycleAreIn2StrongBlock}
Let $G=(V,E)$ be a strongly connected graph and let $v_0,v_1,\ldots,v_l$ be distinct vertices of $G$ such that $v_0 \overset{2s}{\leftrightsquigarrow } v_l$ and $v_{i-1} \overset{2s}{\leftrightsquigarrow } v_i$ for $i \in\lbrace 1,2\ldots,l\rbrace$. Then all the vertices $v_0,v_1,\ldots,v_l$ lie in the same $2$-strong block of $G$.
\end{lemma}
\begin{proof}
Let $v_r,v_q$ be two vertices such that $r,q\in \lbrace 0,1,\ldots,l\rbrace$ and $r<q$. Let $w$ be a vertex in $ V\setminus\lbrace v_r,v_q\rbrace$. We consider two cases.
\begin{enumerate}
\item $w\notin \lbrace v_{r+1},v_{r+2},\ldots,v_{q-1}\rbrace$. Then, for each $i\in \lbrace r+1,r+2,\ldots,q\rbrace$, the vertices $v_{i-1},v_i$ lie in the same SCC of $G\setminus \lbrace w\rbrace$. Thus the vertices $v_r,v_q$ lie in the same SCC of $G\setminus \lbrace w\rbrace$.
\item $w \in \lbrace v_{r+1},v_{r+2},\ldots,v_{q-1}\rbrace$. Then the vertices $v_{i-1},v_i$ lie in the same SCC of $G\setminus \lbrace w\rbrace$ for each $i\in \lbrace 1,2,\ldots,r\rbrace \cup\lbrace q+1,q+2,\ldots,l\rbrace$. Furthermore, the vertices $v_0,v_l$ lie in the same SCC of $G\setminus \lbrace w\rbrace$ since  $v_0 \overset{2s}{\leftrightsquigarrow } v_l$. Thus the vertices $v_r,v_q$ lie in the same SCC of $G\setminus \lbrace w\rbrace$.
\end{enumerate}
Since the vertices $v_r,v_q$ lie in the same SCC of $G\setminus \lbrace w\rbrace$ for any vertex $w\in V\setminus\lbrace v_r,v_q\rbrace$, the vertices $v_r,v_q$ lie in the same $2$-strong block of $G$.
\end{proof}
Algorithm \ref{algo:algor1forall2strongblocks} shows our first algorithm for computing the $2$-strong blocks of a strongly connected graph $G=(V,E)$.
\begin{figure}[htbp]
\begin{myalgorithm}\label{algo:algor1forall2strongblocks}\rm\quad\\[-5ex]
\begin{tabbing}
\quad\quad\=\quad\=\quad\=\quad\=\quad\=\quad\=\quad\=\quad\=\quad\=\kill
\textbf{Input:} A strongly connected graph $G=(V,E)$.\\
\textbf{Output:} The $2$-strong blocks of $G$.\\
{\small 1}\> lines $1$--$5$ of Algorithm \ref{algo:algorforall2dblocks}.\\
{\small 6}\>\> \textbf{for} each vertex $v\in V$ \textbf{do}\\
{\small 7}\>\>\> Compute $DT(v)$.\\
{\small 8}\>\>\> \textbf{for} each direct successor $w$ of $v$ in $DT(v)$ \textbf{do}\\
{\small 9}\>\>\>\> $A[v,w] \leftarrow 1$. \\
{\small 10}\>\> lines $9$--$13$ of Algorithm \ref{algo:algorforall2dblocks}.
\end{tabbing}
\end{myalgorithm}
\end{figure}

Using arguments similar to those in the proof of Lemma \ref{def:Algo1IsCorrect}, one can show that Algorithm \ref{algo:algor1forall2strongblocks} is correct.

\begin{theorem}\label{def:algo2sbrunningtime}
Algorithm \ref{algo:algor1forall2strongblocks} runs in $O(nm)$ time.
\end{theorem} 
\begin{proof}
The dominators of a flowgraph can be found in linear time \cite{BGKRTW00,AHLT99}. Therefore, lines $6$--$9$ take $O(nm)$ time.
\end{proof}
\begin{lemma}\label{def:Relationbetween2sAndSAPs}
Let $G=(V,E)$ be a strongly connected graph and let $x,y$ be distinct vertices in $G$. Let $S$ be the set of all the SAPs in $G$. Then for any vertex $z \in V\setminus (S \cup \lbrace x,y\rbrace)$, the vertices $x$ and $y$ lie in the same SCC of $G\setminus \lbrace z\rbrace$.
\end{lemma}
\begin{proof}
 Immediate from the definition. 
\end{proof}
This simple lemma gives rise to an alternative algorithm (Algorithm \ref{algo:algor2forall2strongblocks}) that might be helpful if the number of the SAPs is small.

\begin{figure}[htbp]
\begin{myalgorithm}\label{algo:algor2forall2strongblocks}\rm\quad\\[-5ex]
\begin{tabbing}
\quad\quad\=\quad\=\quad\=\quad\=\quad\=\quad\=\quad\=\quad\=\quad\=\kill
\textbf{Input:} A strongly connected graph $G=(V,E)$.\\
\textbf{Output:} The $2$-strong blocks of $G$.\\
{\small 1}\> \textbf{if} $G$ is $2$-vertex-connected \textbf{then}\\
{\small 2}\>\> Output $V$.\\
{\small 3}\> \textbf{else}\\
{\small 4}\>\> Let $A$ be an $n\times n$ matrix.\\
{\small 5}\>\> Initialize $A$ with $1$s.\\
{\small 6}\>\> Compute the SAPs of $G$.\\
{\small 7}\>\> \textbf{for} each SAP $s$ of $G$ \textbf{do} \\
{\small 8}\>\>\> Compute the SCCs of $G\setminus \lbrace s\rbrace$.\\
{\small 9}\>\>\> \textbf{for} each pair $(v,w) \in (V \setminus\lbrace s\rbrace)\times (V \setminus\lbrace s\rbrace)$ \textbf{do} \\
{\small 10}\>\>\>\> \textbf{if} $v,w$ in different SCCs of $G\setminus \lbrace s\rbrace$ \textbf{then}\\
{\small 11}\>\>\>\>\> $A[v,w] \leftarrow 0$. \\
{\small 12}\>\> $E^{*} \leftarrow \emptyset$. \\
{\small 13}\>\> \textbf{for} each pair $(v,w) \in V\times V $ \textbf{do} \\
{\small 14}\>\>\> \textbf{if} $A[v,w]=1$ and $A[w,v]=1$ \textbf{then} \\
{\small 15}\>\>\>\> Add the undirected edge $(v,w)$ to $E^{*}$. \\
{\small 16}\>\> Compute the blocks of size $>1$ of $G^{*}=(V,E^{*})$ and output them. 
\end{tabbing}
\end{myalgorithm}
\end{figure}
\begin{lemma}\label{def:vwInSameSCCiFA1}
Let $v,w$ be distinct vertices in a strongly connected graph $G$. Then $v \overset{2s}{\leftrightsquigarrow } w$ if and only if $A[v,w]=1$ and $A[w,v]=1$ (when line $14$ is reached).
\end{lemma}
\begin{proof}
``$\Leftarrow$'' If $A[v,w]=1$ and $A[w,v]=1$, then the vertices $v,w$ lie in the same SCC of $G\setminus \lbrace s\rbrace$ for any SAP $s\in V\setminus\lbrace v,w\rbrace$ (see lines $7$--$11$). By Lemma \ref{def:Relationbetween2sAndSAPs}, the vertices $v,w$ lie in the same SCC of $G\setminus \lbrace z\rbrace$ for any vertex $z\in V\setminus\lbrace v,w\rbrace$.\\
``$\Rightarrow$'' This follows from Lemma \ref{def:Relationbetween2sAndSAPs}.
\end{proof}
\begin{theorem}\label{def:RunTimealgor2forall2strongblocks}
The running time of Algorithm \ref{algo:algor2forall2strongblocks} is $O(t_{sap}n^{2})$.
\end{theorem} 
\begin{proof}
The SAPs of a directed graph can be computed in linear time using the algorithm of Italiano et al. \cite{ILS12}. Lines $7$--$11$ take $O(t_{sap}n^{2})$ time.
\end{proof}
\begin{corollary}\label{def:2strongblocksCanbeComputedInTime}
The $2$-strong blocks of a directed graph $G=(V,E)$ can be computed in $O( \min \lbrace m,t_{sap} n\rbrace  n)$ time.
\end{corollary} 
\section{Computing the $2$-edge blocks} \label{def:sec2eb}
In this section we present two algorithms for computing the $2$-edge blocks of directed graphs. The $2$-edge blocks of a directed graph are the union of the $2$-edge blocks of its SCCs. We define a relation $\overset{2e}{\leftrightsquigarrow } $ as follows. For any distinct vertices $x,y \in V$, we write $x \overset{2e}{\leftrightsquigarrow } y$ if there exist two edge-disjoint paths from $x$ to $y$ and two edge-disjoint paths from $y$ to $x$ in $G$. The $2$-edge blocks are maximal subsets closed under $x \overset{2e}{\leftrightsquigarrow } y$.
\begin{lemma} \label{def:2eblemma1}
Let $G=(V,E)$ be a strongly connected graph and let $x$ and $y$ be distinct vertices in $G$. Then $x \overset{2e}{\leftrightsquigarrow } y$ if and only if for each edge $(v,w)\in E$, the vertices $x,y$ lie in the same SCC of $G\setminus \lbrace (v,w)\rbrace$.
\end{lemma}
\begin{proof} This is an immediate consequence of Menger's Theorem for edge connectivity \cite{BR12}.
\end{proof}
\begin{lemma}\label{def:2ebsAredisjoint}
Let $G=(V,E)$ be a strongly connected graph. The $2$-edge blocks of $G$ are disjoint.
\end{lemma}
\begin{proof}
Let $C^{2e}_1,C^{2e}_2$ be two distinct $2$-edge blocks of $G$. Assume for a contradiction that $C^{2e}_1 \cap C^{2e}_2 \neq \emptyset$. Then there are two vertices $x\in C^{2e}_1,y\in C^{2e}_2$ with $x,y \notin C^{2e}_1 \cap C^{2e}_2$ and an edge $(v,w)\in E$ such that the vertices $x,y$ lie in distinct SCCs of $G\setminus \lbrace (v,w)\rbrace$. Let $z$ be a vertex in $C^{2e}_1 \cap C^{2e}_2$.
Then the vertices $x,z$ lie in the same SCC of $G\setminus \lbrace (v,w)\rbrace$ since $ C^{2e}_1$ is a $2$-edge block
and the vertices $z,y$ lie in the same SCC of $G\setminus \lbrace (v,w)\rbrace$ since $C^{2e}_2$ is a $2$-edge block. Hence $x,y$ lie in the same SCC of $G\setminus \lbrace (v,w)\rbrace$, a contradiction.
\end{proof}
Algorithm \ref{algo:Algor1ForAll2EdgeBlocks} shows our first algorithm for computing the $2$-edge blocks of a strongly connected graph $G$.
\begin{figure}[htbp]
\begin{myalgorithm}\label{algo:Algor1ForAll2EdgeBlocks}\rm\quad\\[-5ex]
\begin{tabbing}
\quad\quad\=\quad\=\quad\=\quad\=\quad\=\quad\=\quad\=\quad\=\quad\=\kill
\textbf{Input:} A strongly connected graph $G=(V,E)$.\\
\textbf{Output:} The $2$-edge blocks of $G$.\\
{\small 1}\> \textbf{if} $G$ is $2$-edge-connected \textbf{then}\\
{\small 2}\>\> Output $V$.\\
{\small 3}\> \textbf{else}\\
{\small 4}\>\> Let $A$ be an $n\times n$ matrix.\\
{\small 5}\>\> Initialize $A$ with $0$s.\\
{\small 6}\>\> \textbf{for} each vertex $v\in V$ \textbf{do}\\
{\small 7}\>\>\> Compute the edge dominators of $G(v)=(V,E,v)$.\\
{\small 8}\>\>\> \textbf{for} each vertex $w\in V\setminus \lbrace v\rbrace$ \textbf{do}\\
{\small 9}\>\>\>\> \textbf{If} there is no edge dominator of $w$ \textbf{then}\\
{\small 10}\>\>\>\>\> $A[v,w] \leftarrow 1$. \\
{\small 11}\>\> $E^{*} \leftarrow \emptyset$. \\
{\small 12}\>\> \textbf{for} each pair $(v,w) \in V\times V$ \textbf{do} \\
{\small 13}\>\>\> \textbf{if} $A[v,w]=1$ and $A[w,v]=1$ \textbf{then} \\
{\small 14}\>\>\>\> Add the undirected edge $(v,w)$ to $E^{*}$. \\
{\small 15}\>\> Compute the connected components of size $>1$ of the graph  \\
{\small 16}\>\> $G^{*}=(V,E^{*})$ and output them. 
\end{tabbing}
\end{myalgorithm}
\end{figure}

Algorithm \ref{algo:Algor1ForAll2EdgeBlocks} works as follows. First, line $1$ tests whether $G$ is $2$-edge-connected, and if it is, line $2$ outputs $V$, since every $2$-edge connected directed graph is a $2$-edge block. Otherwise, for each vertex $v$ in $G$, the algorithm computes the edge dominators of the flowgraph $G(v)=(V,E,v)$, and for each vertex $w\in V\setminus \lbrace v\rbrace$, line $10$ sets $A[v,w]$ to $1$ if there is no edge dominator of $w$. Let $v,w$ be distinct vertices in $G$. Then $v \overset{2e}{\leftrightsquigarrow } w$ if and only if $A[v,w]=1$ and $A[w,v]=1$ in line $13$. Lines $11$--$14$ constructs an undirected graph $G^{*}=(V,E^{*})$ as follows. For each pair $(v,w) \in V\times V$, we add an undirected edge $(v,w)$ to $E^{*}$ if $A[v,w]=1$ and $A[w,v]=1$. Finally, the algorithm finds the connected components of size at least $2$ of $G^{*}$. This is correct by Lemma \ref{def:2ebsAredisjoint}.

In \cite{ILS12}, Italiano et al. presented two algorithms for calculating the strong bridges of a strongly connected graph $G=(V,E)$. We use them to implement lines $8$--$10$ of Algorithm \ref{algo:Algor1ForAll2EdgeBlocks} as follows. Consider a flowgraph $G(v)=(V,E,v)$. For each edge $e=(x,y)\in E$, we delete this edge from $G(v)$ and we add two new edges $(x,\varphi (e)),(\varphi (e),y)$ to $G(v)$. We obtain a new flowgraph, denoted $G'(v)=(V',E',v)$. Then, we compute the dominator tree $DT'(v)$ of $G'(v)$. Obviously, an edge $e$ is an edge dominator of vertex $w \in V\setminus \lbrace v\rbrace$ in $G(v)$ if and only if the corresponding vertex $\varphi (e)$ is a dominator of $w$ in $G'(v)$. We mark  the vertices of $G$ that have edge dominators in $G(v)$ by depth first search in $DT'(v)$. Therefore, lines $8$--$10$ can be implemented in linear time. In \cite{ILS12}, Italiano et al. observed that the strong bridges of $G$ are the SAPs of the directed graph $G'=(V',E')$ that correspond to edges in $G$. We will use these strong bridges in our second algorithm for computing the $2$-edge blocks of $G$.
\begin{theorem}\label{def:Algor1ForAll2EdgeBlocksrunningtime}
Algorithm \ref{algo:Algor1ForAll2EdgeBlocks} runs in $O(nm)$ time.
\end{theorem} 
\begin{proof}
One can test whether a directed graph is $2$-edge-connected in linear time using the algorithm of Italiano et al. \cite{ILS12}. Furthermore, the edge dominators of a flowgraph $G(v)=(V,E,v)$ can be computed in linear time \cite{ILS12,FILOS12}. Lines $6$--$10$ take $O(nm)$ time. The connected components of $G^{*}$ can be found in $O(n^{2})$ time.    
\end{proof} 
\begin{lemma}\label{def:Relationbetween2ebAndSbs}
Let $G=(V,E)$ be a strongly connected graph and let $x,y$ be distinct vertices in $G$. Let $S_{ sb}$ be the set of all the strong bridges of $G$. Then for any edge $e \in E\setminus S_{ sb}$, the vertices $x$ and $y$ lie in the same SCC of $G\setminus \lbrace e\rbrace$.
\end{lemma}
\begin{proof} Immediate from the definition.  
\end{proof}
This simple lemma leads to another algorithm (Algorithm \ref{algo:algor2forall2edgeblocks}) which might be useful when $t_{sb}$ is small. 
\begin{figure}[htbp]
\begin{myalgorithm}\label{algo:algor2forall2edgeblocks}\rm\quad\\[-5ex]
\begin{tabbing}
\quad\quad\=\quad\=\quad\=\quad\=\quad\=\quad\=\quad\=\quad\=\quad\=\kill
\textbf{Input:} A strongly connected graph $G=(V,E)$.\\
\textbf{Output:} The $2$-edge blocks of $G$.\\
{\small 1}\> \textbf{If} $G$ is $2$-edge-connected \textbf{then}.\\
{\small 2}\>\> Output $V$.\\
{\small 3}\> \textbf{else}\\
{\small 4}\>\> Let $A$ be an $n\times n$ matrix.\\
{\small 5}\>\> Initialize $A$ with $1$s.\\
{\small 6}\>\> \textbf{for} each strong bridge $e$ of $G$ \textbf{do} \\
{\small 7}\>\>\> \textbf{for} each pair $(v,w) \in V\times V$ \textbf{do} \\
{\small 8}\>\>\>\> \textbf{if} $v,w$ in distinct SCCs of $G\setminus \lbrace e\rbrace$ \textbf{then}\\
{\small 9}\>\>\>\>\> $A[v,w] \leftarrow 0$. \\
{\small 10}\>\> Lines $12$--$15$ of Algorithm \ref{algo:algor2forall2strongblocks} to construct $G^{*}=(V,E^{*})$. \\
{\small 14}\>\> Compute the connected components of size $>1$ of $G^{*}$ and output them.\\
\end{tabbing}
\end{myalgorithm}
\end{figure}

The correctness of this algorithm follows from the following lemma.
\begin{lemma}\label{def:vwInSame2EdgeBlocks}
Let $v,w$ be distinct vertices in a strongly connected graph $G$. Then  $v \overset{2e}{\leftrightsquigarrow } w$ if and only if $A[v,w]=1$ and $A[w,v]=1$ (when line $10$ is reached).
\end{lemma}
\begin{proof}
Similar to the proof of Lemma \ref{def:vwInSameSCCiFA1} using Lemma \ref{def:Relationbetween2ebAndSbs}.
\end{proof}
\begin{theorem}\label{def:RunTimeAlgor2ForAll2EdgeBlocks}
Algorithm \ref{algo:algor2forall2edgeblocks} runs in $O(t_{sb}n^{2})$ time.
\end{theorem} 
\begin{proof}
 The strong bridges of a directed graph can be found in linear time using the algorithm of Italiano et al. \cite{ILS12}. Lines $7$--$11$ take $O(t_{sb}n^{2})$ time.
\end{proof}
Let $G$ a directed graph. Italiano et al. \cite{ILS12} showed that $t_{sb}\leq(2n-2)$.
\begin{corollary}\label{def:RunTimeAll2EdgegBlocksCanbeComputed}
The $2$-edge blocks of a directed graph $G=(V,E)$ can be computed in $O( \min \lbrace m,t_{sb} n\rbrace  n)$ time.
\end{corollary} 
Now we show that the $2$-edge block that contains a certain vertex can be computed in linear time. Let $G=(V,E)$ be a strongly connected graph and let $v\in V$. By $U(v)$ we denote the set of vertices that do not have edge dominators in $G(v)$ and by $U^{R}(v)$ we denote the set of vertices that do not have edge dominators in $G^{R}(v)$. Let $C^{2e}$ be the $2$-edge block of $G$ that includes $v$. The following lemma shows that $C^{2e}= U(v)\cap U^{R}(v)$.
\begin{lemma} \label{def:2EdgeBlockCertainVertex}
 $w\in C^{2e}$ if and only if $ w\in U(v)\cap U^{R}(v)$
\end{lemma}
\begin{proof}
``$\Leftarrow$'' Let $ w\in (U(v)\cap U^{R}(v))\setminus\lbrace v\rbrace$. $w$ does not have any edge dominator in $G(v)$. Therefore, by Menger's Theorem for edge connectivity, there exist two edge-disjoint paths from $v$ to $w$ in $G(v)$. Furthermore, there are two edge-disjoint paths from $v$ to $w$ in $G^{R}(v)$ since $w$ does not have any edge dominator in $G^{R}(v)$. Thus, there are also two edge-disjoint paths from $w$ to $v$ in $G$.\\
``$\Rightarrow$'' Immediate from definition. 
\end{proof}
We have seen that $U(v)$ can be computed in linear time. Therefore, $U(v)\cap U^{R}(v)$ can be computed in linear time.
\section{The relation between $2$-directed blocks, $2$-strong blocks and $2$-edge blocks} \label{def:secanorithAlgorithmfor2db}
In this section we consider the relation between $2$-directed blocks, $2$-strong blocks and $2$-edge blocks.
\begin{lemma}\label{def:LemmaForSeconedAlgorFor2Edgeblocks}
Let $G=(V,E)$ be a strongly connected graph and let $x,y$ be distinct vertices in $G$. Then $x \overset{2}{\leftrightsquigarrow } y$ if and only if $x\overset{2s}{\leftrightsquigarrow} y$ and $x\overset{2e}{\leftrightsquigarrow} y$.
\end{lemma}
\begin{proof}
``$\Leftarrow$'': By Lemma \ref{def:2eblemma1}, for each edge $e\in E$ the vertices $x,y$ lie in the same SCC of $G\setminus \lbrace e\rbrace$ since $x\overset{2e}{\leftrightsquigarrow} y$. Because the vertices $x,y$ lie in the same SCC of $G\setminus \lbrace (x,y)\rbrace$, there exist a path from $x$ to $y$ in $G\setminus \lbrace (x,y)\rbrace$. There is also a path from to $y$ to $x$ in $G\setminus \lbrace (y,x)\rbrace$ since $x,y$ lie in the same SCC of $G\setminus \lbrace (y,x)\rbrace$. As a consequence, the vertices $x,y$ lie in the same SCC of $G\setminus \lbrace (x,y),(y,x)\rbrace$. By Lemma \ref{def:lem1}, we have $x \overset{2}{\leftrightsquigarrow } y$.\\
``$\Rightarrow$'': This direction follows from Lemma \ref{def:Each2dbIsIn2sb} and Lemma \ref{def:xyInSameSccWithoutE}.
\end{proof}
Now we describe our second algorithm for computing all the $2$-directed blocks of a strongly connected graph $G$. First, we execute lines $1$--$11$ of Algorithm \ref{algo:algor2forall2strongblocks}. Next, we execute lines $6$--$9$ of Algorithm \ref{algo:algor2forall2edgeblocks}. Finally, we execute lines $12$--$16$ of Algorithm \ref{algo:algor2forall2strongblocks}. The correctness of our algorithm follows from Lemma \ref{def:LemmaForSeconedAlgorFor2Edgeblocks}.

\begin{theorem}\label{def:RunTimealgor2forall2directedblocks}
The $2$-directed blocks of a directed graph $G$ can be computed in $O( (t_{sap}+t_{sb})n^{2})$ time.
\end{theorem} 
\begin{proof}
This follows from Theorem \ref{def:RunTimealgor2forall2strongblocks} and Theorem \ref{def:RunTimeAlgor2ForAll2EdgeBlocks}.
\end{proof}
\begin{corollary}\label{def:RunningTimerall2directedblocks}
The $2$-directed blocks of a directed graph $G=(V,E)$ can be computed in $O(\min\lbrace m, (t_{sap}+t_{sb})n\rbrace n)$ time.
\end{corollary} 
\begin{theorem}\label{def:spacerequirements}
All algorithms in Sections \ref{def:c2dbofdg}, \ref{def:secC2sb}, \ref{def:sec2eb} and \ref{def:secanorithAlgorithmfor2db} require $O(n^{2})$ space.
\end{theorem}
\begin{proof}
Clearly, all these algorithms need $O(n^{2})$ space to store the matrix $A$ and the auxiliary graph $G^{*}$.
\end{proof}
\section{The $2$-directed blocks that contain a certain vertex}
Let $G=(V,E)$ be a strongly connected graph and let $v$ be a vertex in $G$. In this section we present an algorithm for computing the $2$-directed blocks of $G$ that contain $v$ in $O(t^{*}m)$ time, where $t^{*}$ is the number of these blocks. This algorithm is based on our Lemmas \ref{def:AllVerticesOfCycleAreIn2directedBlock} and \ref{def:lem2}. It offers two advantages, First, it does not need to construct the auxiliary graph $G^{*}$. Second, it runs in linear time when $v$ is contained in only one $2$-directed block. By $N(v)$ we denote the set of all vertices $w\in V\setminus\lbrace v\rbrace$ such that $v \overset{2}{\leftrightsquigarrow } w$. One can compute $N(v)$ by using Procedure \ref{algo:ProcedureForComputingN} in linear time.
\begin{figure}[htbp]
\begin{Procedure}\label{algo:ProcedureForComputingN}\rm\quad\\[-5ex]
\begin{tabbing}
\quad\quad\=\quad\=\quad\=\quad\=\quad\=\quad\=\quad\=\quad\=\quad\=\kill
\textbf{Input:} A strongly connected graph $G=(V,E)$ and vertex $v\in V$.\\
\textbf{Output:} $N(v)$.\\
{\small 1}\> $N_{1}(v) \leftarrow \emptyset$, $N_{2}(v) \leftarrow \emptyset$, $N(v) \leftarrow \emptyset$.\\
{\small 2}\> $E'\leftarrow E$. \\
{\small 3}\> $V'\leftarrow V$. \\
{\small 4}\> \textbf{for} each edge $e=(v,w)\in E$ \textbf{do}\\
{\small 5}\>\> $E'\leftarrow E'\setminus \lbrace (v,w)\rbrace$.\\
{\small 6}\>\> $V'\leftarrow V'\cup \lbrace u_e\rbrace$.\\
{\small 7}\>\> $E'\leftarrow E' \cup \lbrace (v,u_e),(u_e,w) \rbrace$.\\
{\small 8}\> Compute the dominator tree $DT'(v)$ of the flowgraph $G'(v)=(V',E',v)$.\\
{\small 9}\> \textbf{for} each direct successor $w$ of $v$ in $DT'(v)$ \textbf{do}\\ 
{\small 10}\>\> \textbf{if} $w\in V$ \textbf{then}\\
{\small 11}\>\>\> $N_{1}(v)\leftarrow N_{1}(v) \cup \lbrace w\rbrace$.\\
{\small 12}\> Compute the dominator tree $DT'^{R}(v)$ of $G'^{R}(v)=(V',E'^{R},v)$.\\
{\small 13}\> \textbf{for} each direct successor $w$ of $v$ in $DT'^{R}(v)$ \textbf{do}\\ 
{\small 14}\>\> \textbf{if} $w\in V$ \textbf{then}\\
{\small 15}\>\>\> $N_{2}(v)\leftarrow N_{2}(v) \cup \lbrace w\rbrace$.\\
{\small 16}\>$N(v)\leftarrow N_{1}(v)\cap N_{2}(v)$.
\end{tabbing}
\end{Procedure}
\end{figure} 

The correctness of Procedure \ref{algo:ProcedureForComputingN} follows from Lemma \ref{def:CorrectnessOfproccheck2vdps} and the fact that $w \overset{2}{\rightsquigarrow } v$ in $G$ if and only if $v \overset{2}{\rightsquigarrow } w$ in $G^{R}$. Algorithm \ref{algo:AlgorForAll2dBlocksThatConatainv} shows our algorithm.
\begin{figure}[htbp]
\begin{myalgorithm}\label{algo:AlgorForAll2dBlocksThatConatainv}\rm\quad\\[-5ex]
\begin{tabbing}
\quad\quad\=\quad\=\quad\=\quad\=\quad\=\quad\=\quad\=\quad\=\quad\=\kill
\textbf{Input:} A strongly connected graph $G=(V,E)$ and vertex $v\in V$.\\
\textbf{Output:} The $2$-directed blocks of $G$ that contain $v$.\\
{\small 1}\> \textbf{if} $G$ is $2$-vertex-connected \textbf{then}\\
{\small 2}\>\> Output $V$.\\
{\small 3}\> \textbf{else}\\
{\small 4}\>\> $R\leftarrow N(v)$.\\
{\small 5}\>\> \textbf{while} $R$ is not empty \textbf{do}\\
{\small 6}\>\>\> Choose arbitrarily a vertex $w \in R$.\\
{\small 7}\>\>\> \textbf{output} $(R\cap N(w) )\cup \lbrace v,w\rbrace$.\\
{\small 8}\>\>\>  $R\leftarrow R\setminus ((R\cap N(w))\cup \lbrace w\rbrace)$.
\end{tabbing}
\end{myalgorithm}
\end{figure}

\begin{lemma}\label{def:AlgorForAll2dBlocksThatConatainvIsCorrect}
Algorithm \ref{algo:AlgorForAll2dBlocksThatConatainv} calculates the $2$-directed blocks that include $v$.
\end{lemma}
\begin{proof}
Let $C^{2d}_{1},C^{2d}_{2},\ldots,C^{2d}_{t}$ be the $2$-directed blocks which contain $v$. By Lemma \ref{def:lem2}, these blocks include only the vertex $v$ in common. Thus, $C^{2d}_{1}\setminus\lbrace v\rbrace,C^{2d}_{2}\setminus\lbrace v\rbrace,\ldots,C^{2d}_{t}\setminus\lbrace v\rbrace$ are disjoint. Obviously, $\bigcup _{1\leq i \leq t} (C^{2d}_{2})\setminus\lbrace v\rbrace \subseteq N(v)$. Let $w$ be a vertex in $N(v)$ and Let $C^{2d}$ be the $2$-directed block of $G$ such that $v,w \in  C^{2d}$. It is sufficient to show that $C^{2d}=(N(w) \cap N(v)) \cup\lbrace v,w\rbrace $. Let $x$ be a vertex in $N(w) \cap N(v)$. Since $v \overset{2}{\leftrightsquigarrow } w$, $w \overset{2}{\leftrightsquigarrow } x$ and $x \overset{2}{\leftrightsquigarrow } v$, by Lemma \ref{def:AllVerticesOfCycleAreIn2directedBlock}, the vertices $x,v,w$ lie in the same $2$-directed block of $G$. Conversely, let $x$ be a vertex in $C^{2d}\setminus\lbrace v,w\rbrace$. Since $v \overset{2}{\leftrightsquigarrow } x$ and $w \overset{2}{\leftrightsquigarrow } x$, we have $x\in N(v)$ and $x \in N(w)$.    
\end{proof}
\begin{theorem}\label{def:RunTimeAlgorForAll2dBlocksThatConatainv}
Algorithm \ref{algo:AlgorForAll2dBlocksThatConatainv} runs in $O(t^{*}m)$, where $t^{*}$ is the number of the $2$-directed blocks that contain $v$.
\end{theorem}
\begin{proof}
We have seen that $N(v)$ can be computed in linear time. Furthermore, the number of iterations of the while-loop in lines $5$--$8$ is $t^{*}$. The total time is thus $O(t^{*}m)$. 
\end{proof}
\section{Approximation algorithm for the MS-SAPs Problem}
In this section we show that there is a $17/3$ approximation algorithm for the MS-SAPs problem. In \cite{G11}, Georgiadis presented a linear time $3$-approximation algorithm for the problem of finding a minimum-cardinality $2$-vertex connected spanning subgraph ($2$VCSS) of $2$-vertex-connected directed graphs. This algorithm is based on the works \cite{G10,GT05,ILS12}. We slightly modify this algorithm and combine it with the algorithm of Zhao et al. \cite{ZNI03} in order to obtain a $17/3$ approximation algorithm for the MS-SAPs problem. We first briefly describe Georgiadis algorithm \cite{G11}. Let $G=(V,E)$ be a $2$-vertex-connected directed graph and let $v$ be a vertex in $G$. Menger's Theorem for vertex connectivity \cite{BR12} implies that the flowgraph $G(v)$ has no non-trivial dominators.
In \cite{GT05}, Georgiadis and Tarjan proved that there exist two independent spanning trees of $G(v)$.  Algorithm \ref{algo:Georgiadisalgorihtm} shows the algorithm of Georgiadis \cite{G11}.
\begin{figure}[htbp]
\begin{myalgorithm}\label{algo:Georgiadisalgorihtm}\rm{(from \cite{G11})}\rm\quad\\[-5ex]
\begin{tabbing}
\quad\quad\=\quad\=\quad\=\quad\=\quad\=\quad\=\quad\=\quad\=\quad\=\kill
\textbf{Input:} A $2$-vertex-connected directed graph $G=(V,E)$.\\
\textbf{Output:} A $2$-vertex-connected spanning subgraph $G^{*}$ of $G$.\\
{\small 1}\> Choose arbitrarily a vertex $v \in V$.\\
{\small 2}\> Compute two independent spanning trees $T_{1},T_{2}$ of $G(v)$.\\
{\small 3}\> Compute two independent spanning trees $T_{3},T_{4}$ of $G^{R}(v)$.\\
{\small 4}\> Construct a strongly connected spanning subgraph (SCSS)\\
 {\small 6}\>\> $G^{'}= (V\setminus\lbrace v\rbrace , E^{'})$ of $G\setminus\lbrace v\rbrace$ with $|E^{'}|\leq 2(n-2)$.\\
{\small 7}\>$E^{*}\leftarrow T_{1}\cup T_{2}\cup T_{3}^{R}\cup T_{4}^{R}\cup E^{'}$.\\
{\small 8}\> Output $G^{*}=(V,E^{*})$.
\end{tabbing}
\end{myalgorithm}
\end{figure}

By \cite[Lemma $2$]{G11}, the flowgraphs $(V,T_{1} \cup T_{2},v)$ and $(V,T_{3}\cup T_{4},v)$ have only trivial dominators. Let $w$ be a vertex in $G\setminus\lbrace v\rbrace$. As is well known, it is easy to calculate a SCSS $G^{'}= (V\setminus\lbrace v\rbrace , E^{'})$ of $G\setminus\lbrace v\rbrace$ with $|E^{'}|\leq 2(n-2)$. Just take the union of outgoing branching rooted at $w$ and incoming branching rooted at $w$ (\cite{FJ81,KRY94}). Since $G^{*}\setminus\lbrace v\rbrace$ is strongly connected, the vertex $v$ is not a SAP in $G^{*}$. Therefore, by \cite[Theorem $5.2$]{ILS12} the directed graph $G^{*}$ has no SAPs.
\begin{theorem} \label{def:ApproximatioRatioOfGeorgiadisAlgorithm}
\rm \cite{G11} Algorithm \ref{algo:Georgiadisalgorihtm} has an approximation ratio of $3$ and runs in linear time.
\end{theorem}
Now we modify Georgiadis algorithm \cite{G11} in order to obtain an approximation algorithm for the MS-SAPs problem.
 \begin{figure}[htbp]
\begin{myalgorithm}\label{algo:ModifiedGeorgiadisalgorihtm}\rm\quad\\[-5ex]
\begin{tabbing}
\quad\quad\=\quad\=\quad\=\quad\=\quad\=\quad\=\quad\=\quad\=\quad\=\kill
\textbf{Input:} A strongly connected graph $G=(V,E)$.\\
\textbf{Output:} A SCSS of $G$ with the same SAPs.\\
{\small 1}\> \textbf{if} $G$ is $2$-vertex-connected \textbf{then}\\
{\small 2}\>\> Run Algorithm \ref{algo:Georgiadisalgorihtm} on $G$.\\
{\small 3}\> \textbf{else}\\
{\small 4}\>\> Compute the SAPs of $G$.\\
{\small 5}\>\> \textbf{If} all vertices in $V$ are SAPs \textbf{then}\\
{\small 6}\>\>\> Compute a SCSS $G^{*}=(V,E^{*})$ of $G$ using the Algorithm of    \\
{\small 7}\>\>\> Zhao et al. presented in \cite{ZNI03} and output $G^{*}$.\\
{\small 8}\>\> \textbf{else}\\
{\small 9}\>\>\> Choose a vertex $v \in V$ such that $v$ is not a SAP of $G$.\\
{\small 10}\>\>\> Compute a SCSS $G'=(V\setminus \lbrace v\rbrace,E')$ of $G\setminus\lbrace v\rbrace$ using the Algorithm\\
{\small 11}\>\>\> of Zhao et al. \cite{ZNI03}. \\
{\small 12}\>\>\> Compute two independent spanning trees $T_{1}, T_{2}$ of $G(v)$.\\
{\small 13}\>\>\> Compute two independent spanning trees $T_{3}, T_{4}$ of $G^{R}(v)$.\\
{\small 14}\>\>\>  $E^{*}\leftarrow E' \cup T_{1} \cup T_{2}\cup T_{3}^{R}\cup T_{4}^{R}$.\\
{\small 15}\>\>\> Output $G^{*}=(V,E^{*})$.
\end{tabbing} 
\end{myalgorithm}
\end{figure}

The following lemma shows that the output $G^{*}$ is a feasible solution for the MS-SAPs problem. 
\begin{Lemma} \label{def:ModifiedGeorgiadisalgorihtmIsCorrect}
The output $G^{*}$ is strongly connected, and the directed graphs $G^{*}$ and $G$ have the same SAPs.  
\end{Lemma}
\begin{proof} We need only to show that this lemma is correct when $G$ is not $2$-vertex-connected. Let $G$ be a strongly connected graph which is not $2$-vertex-connected. If all the vertices in $V$ are SAPs in $G$, then, for any SCSS $G'$ of $G$, the graphs $G'$ and $G$ have the same SAPs. Otherwise, $G$ has at least one vertex $v$ which is not a SAP of $G$ (see line $9$--$14$). In this case, by \cite[Theorem $5.2$]{ILS12}, $D(v)\cup D^{R}(v)$ is the set of all SAPs of $G$. Georgiadis and Tarjan \cite{GT05} showed that there exist two independent trees of $G(v)$ and the flowgraphs $G(v)$ and $(V,T_{1}\cup T_{2},v)$ have the same non-trivial dominators. Thus, the flowgraphs $G^{R}(v)$ and $(V,T_{3}\cup T_{4},v)$ have the same non-trivial dominators. Clearly, $(V,T_1\cup T_3^{R})$ is strongly connected. Therefore, $G^{*}$ is strongly connected. Let $D_{1}(v)$ be the set of all non-trivial dominators of $G^{*}(v)$ and let $D_{1}^{R}(v)$ be the set of all non-trivial dominators of $G^{*R}(v)$, where $G^{*R}$ is the reversal graph of $G^{*}$. Since 
$T_{1}\cup T_{2}\subseteq E^{*}$, we have $D_{1}(v) \subseteq D(v)$. Moreover, $D_{1}^{R}(v)\subseteq D^{R}(v)$ because $T_{3}\cup T_{4}\subseteq E^{*}$. As a consequence, $D_{1}(v)\cup D_{1}^{R}(v)\subseteq D(v)\cup D^{R}(v)$. Assume for a contradiction that $D_{1}(v)\cup D_{1}^{R}(v) \neq D(v)\cup D^{R}(v)$. Then there is at least vertex $x\in D(v)\cup D^{R}(v)$ such that $x\notin D_{1}(v)\cup D_{1}^{R}(v)$. By \cite[Theorem $5.2$]{ILS12}, the vertex $x$ is not a SAP in $G^{*}$ but $x$ is a SAP in $G$. Thus, $G^{*}\setminus \lbrace x\rbrace$ is strongly connected.
Because $G^{*}\setminus \lbrace x\rbrace$ is a spanning subgraph of $G\setminus \lbrace x\rbrace$, the directed graph $G\setminus \lbrace x\rbrace$ is strongly connected, which contradicts that $x$ is a SAP in $G$. 
Since $v$ is not a SAP of $G^{*}$ and $D(v)\cup D^{R}(v)=D_{1}(v)\cup D_{1}^{R}(v)$, by \cite[Theorem $5.2$]{ILS12} the set of all SAPs of $G^{*}$ is $D(v)\cup D^{R}(v)$.
\end{proof}
\begin{theorem}\label{def:modifierdAlgorithHasApprimationRatio}
Algorithm \ref{algo:ModifiedGeorgiadisalgorihtm} achieves an approximation ratio of $17/3$.
\end{theorem}
\begin{proof} The algorithm of Zhao et al. \cite{ZNI03} has an approximation factor of $5/3$. Thus, it is enough to consider the case when $G$ is not $2$-vertex-connected and includes at least one vertex which is not a SAP.
Let $E_{opt}$ be an optimal solution for the MS-SAPs problem. Let $E'_{opt}$ be any optimal solution for the problem of finding a MSCSS of $G\setminus\lbrace v\rbrace$. Since the vertex $v$ is not a SAP in $G[E_{opt}]$, the graph $G[E_{opt}]\setminus\lbrace v\rbrace$ is strongly connected. Thus, we have $|E_{opt}|\geq |E'_{opt}|$. Consequently, by \cite[Theorem $3$]{ZNI03}, we have $|E'|/|E_{opt}|\leq E'|/|E'_{opt}|\leq 5/3$.
Since $|T_{1} \cup T_{2}\cup T_{3}^{R}\cup T_{4}^{R}| \leq 4(n-1)$ and $|E_{opt}|\geq n$, Algorithm \ref{algo:ModifiedGeorgiadisalgorihtm} has approximation ratio $|E^{*}|/|E_{opt}|\leq 4-4/n+5/3=17/3-4/n$. 
\end{proof}
\begin{theorem}\label{def:RunningTimeModifiedGeorgiadisalgorihtm}
Algorithm \ref{algo:ModifiedGeorgiadisalgorihtm} runs in linear time. 
\end{theorem}
\begin{proof}
The SAPs of $G$ can be calculated in linear time using the algorithm of Italiano et al \cite{ILS12}. Two independent spanning trees of $G(v)$ can also be constructed in linear time \cite{GT12}. Furthermore, the algorithm of Zhao et al. \cite{ZNI03} can be implemented in linear time.
\end{proof}
\section{Approximation algorithm for the MS-2SBs problem}
In this section we present an approximation algorithm for the MS-2SBs problem. Let $G=(V,E)$ be a strongly connected graph such that $G$ is not $2$-vertex-connected. Our algorithm consists of two main steps. The first step finds a SCSS $G^{*}$ of $G$ such that $G^{*}$ and $G$ have the same SAPs. The following lemma explains the purpose of this step.
\begin{lemma}\label{def:MSCSS2SBsFirstLemma}
Let $G=(V,E)$ be a strongly connected graph and let $S$ be the set of all SAPs of $G$.
 Let $G'=(V,E')$ be a feasible solution for the MS-SAPs problem and let $x,y$ be distinct vertices in $V$. Then for any vertex $z\in V\setminus (S\cup\lbrace x,y\rbrace)$, the vertices $x,y$ lie in the same SCC of $G'\setminus\lbrace z\rbrace$.
 \begin{proof}
 Immediate from the definition of SAPs, since $G$ and $G'$ have the same SAPs.
 \end{proof}
\end{lemma} 
Let $x,y$ be two vertices in $G$ such that $x \overset{2s}{\leftrightsquigarrow } y$ and let $z$ be a vertex in $V\setminus \lbrace x,y\rbrace$ such that $z$ is not SAP in $G$. The first step ensures that there exist at least one path from $x$ to $y$ and from $y$ to $x$ in $G^{*}\setminus \lbrace z\rbrace$.
The second step computes, for each SAP $v$ of $G$, strongly connected spanning subgraphs of the subgraphs induced by the SCCs of $G\setminus\lbrace v\rbrace$. Algorithm \ref{algo:AprroximationAlgoForMSCSS2SBs} describes our algorithm.
  \begin{figure}[htbp]
\begin{myalgorithm}\label{algo:AprroximationAlgoForMSCSS2SBs}\rm\quad\\[-5ex]
\begin{tabbing}
\quad\quad\=\quad\=\quad\=\quad\=\quad\=\quad\=\quad\=\quad\=\quad\=\kill
\textbf{Input:} A strongly connected graph $G=(V,E)$.\\
\textbf{Output:} A SCSS of $G$ with the same $2$-strong blocks.\\
{\small 1}\> \textbf{if} $G$ is $2$-vertex-connected \textbf{then}\\
{\small 2}\>\> Run algorithm of Cheriyan and Thurimella \cite{CT00} for minimum \\
{\small 3}\>\>  cardinality $2$-VCSS problem, improved in \cite{G11}.\\
{\small 4}\> \textbf{else}\\
{\small 5}\>\> lines $4$--$14$ of Algorithm \ref{algo:ModifiedGeorgiadisalgorihtm}, giving $G^{*}=(V,E^{*})$.\\
{\small 6}\>\> \textbf{for} each SAP $v$ of $G$ \textbf{do}\\
{\small 7}\>\>\> Compute the SCCs of $G\setminus\lbrace v\rbrace$.\\
{\small 8}\>\>\> \textbf{for} each SCC $C$ of $G\setminus\lbrace v\rbrace$ \textbf{do}\\
{\small 9}\>\>\>\> \textbf{if} $G^{*}[C]$ is not strongly connected \textbf{then}\\
{\small 10}\>\>\>\>\> choose a vertex $w\in C$. \\
{\small 11}\>\>\>\>\> Construct a spanning out-tree $T_{1}$ of $G[C]$ rooted at $w$.\\
{\small 12}\>\>\>\>\> Construct a spanning out-tree $T_{2}$ of $G^{R}[C]$ rooted at $w$.\\
{\small 13}\>\>\>\>\> $E^{*}\leftarrow E^{*}\cup T_{1} \cup T_{2}^{R}$.\\
{\small 14}\>\> Output $G^{*}=(V,E^{*})$.
\end{tabbing} 
\end{myalgorithm}
\end{figure}

In the following lemma we show that Algorithm \ref{algo:AprroximationAlgoForMSCSS2SBs} returns a feasible solution for the MS-2SBs problem. 
\begin{lemma}\label{def:AprroximationAlgoForMSCSS2SBsIsCorrect}
Let $G^{*}$ be the output of Algorithm \ref{algo:AprroximationAlgoForMSCSS2SBs}. Then $G^{*},G$ have the same $2$-strong blocks and $G^{*}$ is strongly connected. 
\end{lemma}
\begin{proof}
Since each $2$-vertex-connected graph is a $2$-strong block, the algorithm of Cheriyan and Thurimella \cite{CT00} returns a feasible solution when $G$ is $2$-vertex-connected. Let $G=(V,E)$ be a strongly connected graph such that $G$ is not $2$-vertex-connected. By Lemma \ref{def:ModifiedGeorgiadisalgorihtmIsCorrect}, the directed graph $G^{*}$ which is calculated in line $5$ and the directed graph $G$ have the same SAPs, and $G^{*}$ is strongly connected. Therefore, the output $G^{*}$ of Algorithm \ref{algo:AprroximationAlgoForMSCSS2SBs} and $G$ have the same SAPs. Obviously, it is sufficient to show the following. Let $x,y\in V$ be distinct vertices such that $x \overset{2s}{\leftrightsquigarrow } y$ in $G$. We must show that $x \overset{2s}{\leftrightsquigarrow } y$ in the output $G^{*}$ of Algorithm \ref{algo:AprroximationAlgoForMSCSS2SBs}. Let $v\in V\setminus\lbrace x,y\rbrace$ be some vertex. By Lemma \ref{def:MSCSS2SBsFirstLemma}, we may assume that $v$ is a SAP. Then $x,y$ lie in the same SCC of $G\setminus\lbrace v\rbrace$. The execution of the loop in lines $6$--$13$ for $v$ enforces that $x,y$ are also in the same SCC of $G^{*}\setminus\lbrace v\rbrace$.
\end{proof}
\begin{theorem}\label{def:ApprRatioOfAlgoForMSCSS2SBs}
Algorithm \ref{algo:AprroximationAlgoForMSCSS2SBs} has an approximation factor of $(2t_{sap}+17/3)$.
\end{theorem}
\begin{proof}
If $G$ is $2$-vertex-connected, the algorithm of Cheriyan and Thurimella \cite{CT00} for the minimum cardinality $2$-VCSS problem achieves an approximation ratio of $1.5$. We consider now the case when $G$ is not $2$-vertex-connected. Let $E_{opt}$ be an optimal solution for the MS-2SBs problem. The output $G^{*}$ of Algorithm \ref{algo:AprroximationAlgoForMSCSS2SBs} consists of two edge sets $E_{1},E_{2}$, where the edge set $E_{1}$ is computed in line $5$ and the edge set $E_{2}$ is computed in lines $6$--$13$. By Theorem \ref{def:modifierdAlgorithHasApprimationRatio}, $|E_{1}|/|E_{opt}|\leq 17/3$. The number of iterations of the for-loop in lines $6$--$13$ is $t_{sap}$. Because the SCCs of a directed graph are disjoint, we have $|E_{2}|< 2t_{sap}n$. Since $|E_{opt}|\leq n$, we have $|E_{2}|/|E_{opt}|<2t_{sap}$.
\end{proof}
\begin{theorem}\label{def:RunTimeOfAlgoForMSCSS2SBs}
Algorithm \ref{algo:AprroximationAlgoForMSCSS2SBs} runs in $O(m(\sqrt{n}+t_{sap})+n^{2})$ time.
\end{theorem}
\begin{proof}
We obtain the time bound by combining known time bounds in the obvious way. The algorithm of Cheriyan and Thurimella \cite{CT00} for the minimum cardinality $2$-VCSS problem has running time $O(m^{2})$. In $2011$, Georgiadis \cite{G11} improved it to $O(m\sqrt{n}+n^{2})$. By Theorem \ref{def:RunningTimeModifiedGeorgiadisalgorihtm}, line $5$ takes $O(n+m)$ time. The SCCs of a directed graph can be found in linear time using Tarjan's algorithm \cite{T72}. Thus, lines $6$--$13$ take $O(t_{sap}m)$ time. 
\end{proof}
Notice that in lines $9$--$13$ of Algorithm \ref{algo:AprroximationAlgoForMSCSS2SBs}, every SCC $C$ which does not contain any vertex of the $2$-strong blocks of $G$ can be safely disregarded.
\section{Approximation algorithm for the MS-$2$EBs problem}
In this section we present an approximation algorithm for the MS-$2$EBs problem. The idea of this algorithm (Algorithm \ref{algo:AprroAlgoForMSCSS2EBs}) is similar to the idea of Algorithm \ref{algo:AprroximationAlgoForMSCSS2SBs}. 
\begin{figure}[h]
\begin{myalgorithm}\label{algo:AprroAlgoForMSCSS2EBs}\rm\quad\\[-5ex]
\begin{tabbing}
\quad\quad\=\quad\=\quad\=\quad\=\quad\=\quad\=\quad\=\quad\=\quad\=\kill
\textbf{Input:} A strongly connected graph $G=(V,E)$.\\
\textbf{Output:} A SCSS of $G$ with the same $2$-edge blocks.\\
{\small 1}\>\> Choose a vertex $v$ of $G$. \\
{\small 2}\>\> Compute two spanning trees $T_{1},T_{2}$ of $G(v)$ (rooted at $v$) such  \\
{\small 3}\>\> that $T_{1},T_{2}$ have only the edge dominators of $G(v)$ in common.\\
{\small 4}\>\> Compute two spanning trees $T_{3},T_{4}$ of $G^{R}(v)$ (rooted at $v$) such  \\
{\small 5}\>\> that $T_{3},T_{4}$ have only the edge dominators of $G^{R}(v)$ in common.\\
{\small 6}\>\> $E^{*}\leftarrow T_{1}\cup T_{2}\cup T_{3}^{R} \cup T_{4}^{R}$.\\
{\small 7}\>\> Find all the strong bridges in $G$.\\
{\small 8}\>\> \textbf{for} each strong bridge $e$ of $G$ \textbf{do}\\
{\small 9}\>\>\> Compute the SCCs of $G\setminus\lbrace e\rbrace$.\\
{\small 10}\>\>\> \textbf{for} each SCC $C$ of $G\setminus\lbrace e\rbrace$ \textbf{do}\\
{\small 11}\>\>\>\> \textbf{if} $G^{*}[C]$ is not strongly connected \textbf{then}\\
{\small 12}\>\>\>\>\> Find a SCSS $(C,E')$ of $G[C]$ with $|E'|<2|C|$. \\
{\small 13}\>\>\>\>\> $E^{*}\leftarrow E^{*}\cup E' $.\\
{\small 14}\>\> Output $G^{*}=(V,E^{*})$.
\end{tabbing} 
\end{myalgorithm}
\end{figure}
\begin{lemma}\label{def:GOutPutHaveSameEdgeDominator}
Let $G^{*}$ be the output of Algorithm \ref{algo:AprroAlgoForMSCSS2EBs}. Then $G^{*}$ is strongly connected and the directed graphs $G^{*},G$ have the same strong bridges.
\end{lemma}
\begin{proof}
Since $T_{1} \cup T_{3}^{R} \subseteq E^{*}$, the graph $G^{*}$ is strongly connected. Tarjan \cite{T76} proved that there exist two spanning trees (rooted at $v$) of $G(v)$ that have only the edge dominators of $G(v)$ in common and he gave algorithms for computing them. Italiano et al. \cite{ILS12} showed that edge $e\in E$ is strong bridge if and only if $e$ is an edge dominator in $G(v)$ or in $G^{R}(v)$. Therefore, the directed graphs $G,(V,T_{1}\cup T_{2}\cup T_{3}^{R} \cup T_{4}^{R})$ have the same strong bridges. Since $(V,T_{1}\cup T_{2}\cup T_{3}^{R} \cup T_{4}^{R})$ is a subgraph of $G^{*}$ and $G^{*}$ is a subgraph of $G$, the directed graphs $G^{*}, G$ have the same strong bridges
\end{proof}
Next we show that Algorithm \ref{algo:AprroAlgoForMSCSS2EBs} outputs a feasible solution for the MS-$2$EBs problem.
\begin{lemma} \label{def:AprroAlgoForMSCSS2EBsIsCorrect}
 Let $G^{*}=(V,E^{*})$ be the output of Algorithm \ref{algo:AprroAlgoForMSCSS2EBs}. Then $G^{*},G$ have the same $2$-edge blocks.
\end{lemma}
\begin{proof}
 Let $x,y$ be distinct vertices such that $x \overset{2e}{\leftrightsquigarrow } y$ in $G$. We must show that $x \overset{2e}{\leftrightsquigarrow } y$ in $G^{*}$. By Lemma \ref{def:2eblemma1}, we need to show that $x,y$ lie in the same SCC of $G^{*}\setminus\lbrace e\rbrace$ for any edge $e \in E^{*}$.
 Let $e$ be an edge in $G^{*}$. We consider two cases.
 \begin{description}
 \item[$1.$] $e$ is not a strong bridge in $G^{*}$. By lemma \ref{def:GOutPutHaveSameEdgeDominator}, $G^{*}$ is strongly connected. Hence by definition of strong bridges the vertices $x,y$ lie in the same SCC of $G^{*}\setminus\lbrace e\rbrace$. 
 \item[$2.$] $e$ is a strong bridge in $G^{*}$. By Lemma \ref{def:GOutPutHaveSameEdgeDominator}, $G$ and $G^{*}$ have the same strong bridges. Since $x,y$ lie in the same SCC of $G\setminus \lbrace e\rbrace$, the execution of the loop in lines $8$--$13$ enforces that the vertices $x,y$ are also in the same SCC of $G^{*}\setminus\lbrace e\rbrace$.
 \end{description}
\end{proof}
\begin{theorem}\label{def:ApproRatioAprroAlgoForMSCSS2EBs}
Let $G^{*}=(V,E^{*})$ be the output of Algorithm \ref{algo:AprroAlgoForMSCSS2EBs}. Then $|E^{*}|<(4+2t_{sb})n$.
\end{theorem}
\begin{proof}
The number of iterations of the for-loop in lines $8$--$13$ is $t_{sb}$. Therefore, we have  $|E^{*}|=|T_{1}\cup T_{2}\cup T_{3}^{R} \cup T_{4}^{R}|+|E^{*}|-|T_{1}\cup T_{2}\cup T_{3}^{R} \cup T_{4}^{R}|\leq 4(n-1)+2t_{sb}(n-1)<(4+2t_{sb})n$.
\end{proof}
Let $G$ be a strongly connected graph. Since each SCSS of $G$ has at least $n$ edges, Algorithm \ref{algo:AprroAlgoForMSCSS2EBs} achieves an approximation ratio of $(4+2t_{sb})$.
\begin{theorem}\label{def:RunTimeAprroAlgoForMSCSS2EBs}
The running time of Algorithm \ref{algo:AprroAlgoForMSCSS2EBs} is $O((t_{sb}+\alpha(n,m))m)$.
\end{theorem}
\begin{proof}
Two spanning trees $T_{1},T_{2}$ of $G(v)$ (rooted at $v$) such that $T_{1},T_{2}$ have only the edge dominators of $G(v)$ in common can be computed in $O(m\alpha(n,m))$ time by using Tarjan's algorithm \cite{T76}, where $\alpha(n,m)$ is a very slowly function related to a functional inverse of Ackermann's function. The strong bridges of a strongly connected graph can be computed in linear time using the algorithm of Italiano et al. \cite{ILS12}. Moreover, the number of iterations of the for-loop in lines $6$--$13$ is $t_{sb}$. The total time for Algorithm \ref{algo:AprroAlgoForMSCSS2EBs} is therefore $O((t_{sb}+\alpha(n,m))m)$.
\end{proof}
 Notice that by Lemma \ref{def:LemmaForSeconedAlgorFor2Edgeblocks}, we can obtain a $(2(t_{sap}+t_{sb})+29/3)$ approximation algorithm for the MS-2DBs problem by combining Algorithm \ref{algo:AprroximationAlgoForMSCSS2SBs} and Algorithm \ref{algo:AprroAlgoForMSCSS2EBs}. This algorithm whose running time is $O((t_{sap}+t_{sb}+\sqrt{n}+\alpha(n,m))m+n^{2})$ might be useful when $t_{sap}+t_{sb}$ is small.

\section{Open problems}
 The $k$-strong blocks of directed graphs which are natural generalization of $2$-strong blocks are similar to the $k$-blocks of undirected graphs \cite{CDHH13}. Let $G=(V,E)$ be a directed graph. We define a relation $\overset{ks}{\leftrightsquigarrow } $ as follows. For any pair of distinct vertices $x,y \in V$, we write $x \overset{ks}{\leftrightsquigarrow } y$ if for each subset  $X \subseteq V\setminus \lbrace x,y\rbrace$ with $|X|<k$, the vertices $x$ and $y$ lie in the same SCC of $G\setminus X$. A \textit{$k$-strong block} in a directed graph $G$ is a maximal vertex set $C^{ks}\subseteq V$ with $|C^{ks}|\geq k$ such that for each pair of distinct vertices $x,y \in C^{ks}$, $x \overset{ks}{\leftrightsquigarrow } y$. One can show that any two $k$-strong blocks have at most $k-1$ vertices in common.
A simple algorithm was given in \cite{CDHH13} to find the $k$-blocks of an undirected graph. We noticed that this algorithm is also able to compute the $k$-strong blocks of a directed graph in $O(\min\lbrace k,\sqrt{n}\rbrace n^{4})$-time. We just need to modify the pre-processing step and the definition of separation. Let $G=(V,E)$ be a directed graph. An ordered pair $(C,D)$ such that $C,D \subseteq V$ and $C\cup D=V$ is a separation of $G$ if there is no edge from $C\setminus D$ to $D\setminus C$ or there is no edge from $D\setminus C$ to $C\setminus D$. In the pre-processing step we construct a new undirected graph $H_{k}=(V,E_{k})$ as follows. For each pair of distinct vertices $x,y $ of $G$, if $x \overset{ks}{\leftrightsquigarrow } y$ in $G$, we add an undirected edge $(x,y)$ to $E_{k}$ (see \ref{def:LemmasRelatedKSBs}). Furthermore, for every pair $(x,y)$ of $H_{k}$ with $(x,y) \notin E_{k}$, we label it with some separation $(C,D)$ such that $|C\cap D|<k$ and $x\in C, y\in D$.  

We leave as an open problem whether there exists any approximation algorithm for the problem of finding MSCSS with same $k$-strong blocks of a strongly connected graph for $k>2$. Another open question is whether there is an approximation algorithm for the problem of finding MSCSS with the same $k$-directed blocks when $k>2$.

It is also possible to generalize the MS-2EBs problem. Let $G=(V,E)$ be a strongly connected graph. We define a relation $\overset{ke}{\leftrightsquigarrow } $ as follows. For any pair of distinct vertices $x,y \in V$, we write $x \overset{ke}{\leftrightsquigarrow } y$ if for each edge subset $Y\subseteq E$ with $|Y|<k$, the vertices $x,y$ lie in the same SCC of $G\setminus Y$.
The $k$-edge blocks of $G$ are maximal subsets of $V$ of size $\geq k$ closed under $\overset{ke}{\leftrightsquigarrow } $. 
\begin{lemma}\label{def:kedgeBlocksAreDisjoint}
The $k$-edge blocks of a strongly connected graph are disjoint.
\end{lemma}
\begin{proof}
The proof is similar to our proof of Lemma \ref{def:2ebsAredisjoint}.
\end{proof}
The $k$-edge blocks can be computed in $O(n^{4}\sqrt{m})$ using flow algorithms since the relation $ \overset{ke}{\leftrightsquigarrow } $ is symmetric and transitive.
\section*{Acknowledgements.}
The author would like to thank Martin Dietzfelbinger for helpful comments and interesting discussions.

\appendix
\section{Lemmas related to $k$-strong blocks} \label{def:LemmasRelatedKSBs}

\begin{Lemma}\label{def:xyLieInKsRelatio}
Let $G=(V,E)$ be a directed graph and let $x,y$ be distinct vertices of $G$. Then $x \overset{ks}{\leftrightsquigarrow } y$ if and only if $x,y$ satisfy one of the following conditions.
\begin{description}
\item[$1.$] $\lbrace (x,y),(y,x) \rbrace \subseteq E$.
\item[$2.$] $(x,y) \in E, (y,x)\notin E$ and there are $k$-vertex-disjoint paths from $y$ to $x$ in $G$.
\item[$3.$]$(y,x) \in E, (x,y)\notin E$ and there are $k$-vertex-disjoint paths from $x$ to $y$ in $G$.
\item[$4.$] $\lbrace (x,y),(y,x) \rbrace \cap E=\emptyset$ and there exist $k$-vertex-disjoint paths from $x$ to $y$ and from $y$ to $x$ in $G$.
\end{description}
\end{Lemma}
\begin{proof}
This follows immediately from Menger's Theorem for vertex connectivity.
\end{proof}
\begin{lemma}\label{def:separationKstrongBlocks}
Let $(C,D)$ be a separation of a directed graph $G=(V,E)$ such that $|C \cap D|< k$. Then $C\setminus D$ and $D\setminus C$ lie in different $k$-strong blocks of $G$.
\end{lemma}
\begin{proof}
Let $x$ be any vertex in $C\setminus D$ and let $y$ be any vertex in $D\setminus C$. By the definition of separation, there is either no path from $x$ to $y$ or no path from $y$ to $x$ in $G\setminus (C\cap D)$. Thus, $x,y$ do not lie in the same SCC of $G\setminus (C\cap D)$.
\end{proof}

\end{document}